\documentclass[
    runinaddress,
    11pt,amsmath,amssymb,
]{revtex4-2}

\usepackage[polish,english]{babel}
\usepackage[utf8]{inputenc}
\usepackage[T1]{fontenc}
\usepackage{graphicx}
\graphicspath{{./}}
\usepackage{subcaption}
\usepackage{dcolumn}
\usepackage{bm}
\usepackage{hyperref}
\usepackage{amsmath}
\usepackage{amsthm}
\usepackage{multirow}
\usepackage{booktabs}
\usepackage{pict2e}
\usepackage{mathrsfs}
\usepackage{rotating}
\usepackage{setspace}

\usepackage{microtype} 
\begin{document}
\singlespacing

\theoremstyle{remark}
\newtheorem*{problem}{Problem}

\theoremstyle{plain}
\newtheorem{theorem}{Theorem}
\newtheorem{lemma}[theorem]{Lemma}
\newtheorem{corollary}[theorem]{Corollary}

\title{Prolonging The Inevitable: Maximising survival time of an engine-equipped spacecraft between spatial hypersurfaces, as applied to the Schwarzschild spacetime}

\author{Karol Urba\'nski}
\email{karol.j.urbanski@doctoral.uj.edu.pl}
\affiliation{Szkoła Doktorska Nauk Ścisłych i Przyrodniczych, and Institute of Physics, Jagiellonian University in Kraków.}
\begin{abstract}
    The fate of an astronaut unfortunate -- or foolish -- enough to find themselves hurtling towards spaghettification after passing the event horizon of a black hole is a common anecdote told by scientists to the regular population. However, despite the fact the Schwarzschild spacetime has been discovered over a century ago, the simple question of how long can such a space traveller live has not been fully elaborated on since. In fact, a few textbooks even give a mistaken or easily misread description of what happens. We address those inconsistencies. We calculate the proper time a space traveller equipped with means of propulsion can expect to live in these circumstances, giving analytical expressions (as elliptic integrals) wherever possible. We prove a principle that explains the best strategy to extend their life, and show its' generalisation for other spacetimes. Finally, we give quantitative answers to what gains due to optimal control can be expected in typical and somewhat `realistic' circumstances.
\end{abstract}

\maketitle

\section{\label{sec:past}Overview of scenario and previous literature}
\subsection{Preliminary information}
Almost immediately after Albert Einstein's groundbreaking discovery of the field equations of gravitation \cite{einstein1915field} that now bear his name, Karl Schwarzschild published \cite{schwarzschild1916gravitationsfeld} the first exact solution of these field equations -- a metric describing a vacuum spacetime with spherical symmetry that is identifiable with the gravitational field outside a mass $M$ in the absence of an electric charge:
\begin{equation}\label{eq:schw_line_element}
    ds^2 = - \left(1-\frac{2M}{r}\right) dt^2 + \left(1-\frac{2M}{r}\right)^{-1} dr^2 + r^2 d\Omega^2 
\end{equation}
where we've assumed geometric units $c = G = 1$, the metric signature $(-,+,+,+)$ and where $d\Omega^{2} = d\theta^{2} + \sin^{2}\theta d\phi^{2}$ is the metric line element of the two sphere. The coordinate ranges were originally $t \in (- \infty, + \infty), r > 2M, \theta \in (0, \pi), \phi \in (0, 2\pi)$.

At first, this description was considered valid only at a radius exceeding the value where the metric becomes singular $r = 2M$. However, as shown by many different authors researching many different phenomena (eg.\ Gullstrand-Painlev{\'e} coordinates \cite{painleve1921mecanique, gullstrand1922allgemeine}, Eddington-Finkelstein coordinates \cite{eddington1924comparison, finkelstein1958past}, Lema{\^\i}tre coordinates \cite{lemaitre1931expanding}, Kruskal-Szekeres coordinates \cite{kruskal1960maximal, szekeres1960singularities}, Synge extension \cite{synge1950gravitational}) it is possible to transform into a coordinate system in which the metric is regular everywhere except $r = 0$. It is also possible to show (first done by Robertson \cite{goenner1998expanding}) that a falling observer will cross the coordinate singularity at $r = 2M$ in finite proper time. In this paper, we will occassionally refer to Eddington-Finkelstein coordinates:

\begin{equation}\label{eq:ef_line_element}
    ds^2 = - \left(1-\frac{2M}{r}\right) dt^2 + \frac{4M}{r} dtdr + \left(1+\frac{2M}{r}\right) dr^2 + r^2 d\Omega^2
\end{equation}
with the same ranges for each coordinate, and Kruskal-Szekeres coordinates describing the maximally extended coordinates:
\begin{equation}\label{eq:ks_line_element}
    ds^2 = - \frac{32 M^3}{r} e^{-\frac{r}{2M}}(-dT^2+dX^2) + r^2 d\Omega^2,
\end{equation}
with coordinates ranging in $- \infty < X < \infty$, $- \infty < T^2 - X^2 < 1$.

The singularity at $r = 0$ appears regardless of the coordinate system used to describe the spacetime, and is a mathematical idealisation of a curvature singularity -- as can be verified by computing the Ricci scalar. Once a massive observer (or a photon) falls below the event horizon at $r = 2M$, it will inexorably reach the singularity \cite{misner1973gravitation}.

As general relativity evolved, it became clear that nothing in the theory prohibits the formation of an object described by \eqref{eq:ef_line_element}, \eqref{eq:ks_line_element} or \eqref{eq:schw_line_element}. The physical existence of black hole solutions was debated for a long time, to the point that in 1974 Stephen Hawking and Kip Thorne famously made a bet about whether convincing observational evidence would be found \cite{hawking1988}. Since then, enough evidence (for rotating black holes) been shown for the bet to be settled in the affirmative, culminating in imaging the shadow of the M87 black hole \cite{akiyama2021first}. A thorough history of the early history of GR can be found in the book \cite{goenner1998expanding}.

\begin{figure}[h]
    \includegraphics[width=0.9\linewidth]{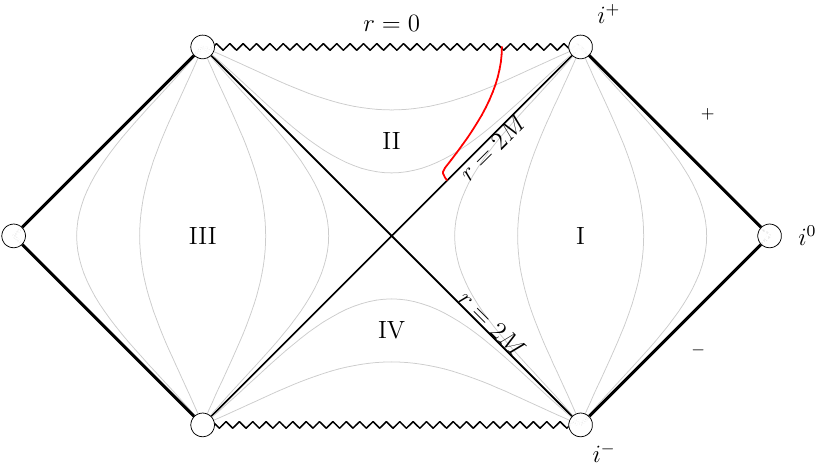}
    \caption{Carter-Penrose diagram for the maximally extended Schwarzschild spacetime. Region I is the asymptotically flat region above the event horizon, described with \eqref{eq:schw_line_element}. Region II is the region under the event horizon, described with \eqref{eq:schw_line_element_under_horizon}, and the subject of this paper. Regions III and IV are the analytic continuation, with the whole spacetime described by \eqref{eq:ks_line_element}; region III is the other asymptotically flat region in the analytic continuation, and region IV is commonly called a 'white hole' spacetime. The red worldline is a schematic example of an accelerated worldline which starts with some $e_0$ at the event horizon but, through acceleration by means of a rocket, soon starts travelling along a line of constant $t$.}
    \label{fig:carter_penrose}
\end{figure}      

The Schwarzschild metric in the region $0 < r < 2 M$ is a special case of the cosmological Kantowski-Sachs model \cite{kantowski1966some}, meaning that results obtained under the horizon of the Schwarzschild black hole are of interest in other contexts.

\subsection{The ``stop squirming'' fallacy}
With the reality of the existence of black holes firmly established, they have become an interesting feature of our current theories of nature in the eyes of laymen and fledgling physicists alike. Many scientific outreach efforts are therefore focused on particular properties of those compact objects. One natural question that arises is ``what would happen to me if I fell into a black hole?'' Indeed, one can find many references commenting on the fate of an unfortunate traveller on such a journey: the description of spaghettification by tidal forces \cite{hawking2009brief, misner1973gravitation} and how below the horizon, what was previously a spatial coordinate becomes a timelike coordinate.

However, on the topic of the amount of time that such a traveller would have to live, the textbooks are somewhat more elusive. Some directly show that the areal radial coordinate $r$ must decrease under the horizon, establishing the inevitability of reaching the singularity \cite{misner1973gravitation}. In others \cite{hartle2003gravity, carroll2004spacetime} an upper bound of proper time $\tau = M \pi$ is calculated. A select few comment further, and it is here where a misconception shows up.

In the excellent book \textit{Spacetime and Geometry} \cite{carroll2004spacetime}, we find the following quote: ``[...] once you enter the event horizon, you are utterly doomed. This is worth stressing; not only can you not escape back to [region above the event horizon], you cannot even stop yourself from moving in the direction of decreasing r, since this is simply the timelike direction. [...] \textbf{Since proper time is maximized along a geodesic, you will live the longest if you don't struggle, but just relax as you approach the singularity.}'' (emphasis added). Similarly, in the solution manual for the book \textit{Gravity} \cite{hartle2003gravity} we can find as an aside to a solution that shows the existence of an upper bound of proper time: ``One of the author's students characterized this result as `The more you struggle, the shorter your life.' ''

While these explanations are evocative, they are either misleading (in the case of \textit{Gravity}'s solutions manual), or wrong (in \textit{Spacetime and Geometry}'s case). The geodesic that reaches the aforementioned upper bound $\tau = M \pi$ is one that no astronaut crossing the event horizon will actually have. Astronauts falling in with some energy and angular momentum will (as we will show in this paper) have a shorter travel time towards the singularity. And while a timelike geodesic certainly maximizes proper time between two events in spacetime, the final destination of our journey (the singularity for $r$ tending to $0$) is not a single event, but a degenerate hypersurface! When we apply thrust using a propulsive rocket of some sort, we will deviate from our initial geodesic connecting $(t_0, 2M, \theta_0, \phi_0)$ and $(t_1, r \rightarrow 0, \theta_1, \phi_1)$. Since the new terminal event $(t_2, r \rightarrow 0, \theta_2, \phi_2)$ can be completely different, the curve that takes us there is no longer restricted from having longer total proper time than our initial geodesic.

This issue has previously been noticed in the paper \cite{lewis2007no}. In it, it is shown that with proper application of thrust in the case of radial infall, we can extend our time to live. However, that work stops at simply showing this fact numerically. In this paper, we will carry on this work, and show exactly what an astronaut equipped with a rocket engine should do; then, we will generalise the solution to apply to a broader class of spacetimes.

\subsection{Precise problem statement}\label{problem}
\begin{problem}
An astronaut, through an unfortunate accident or a grisly death wish, crosses the event horizon of a $M$ mass Schwarzschild black hole. At $r = 2M$, they have initial energy per unit mass $e_0$ and angular momentum per unit mass $L_0$. Their spacecraft is equipped with an engine capable of generating a 3-acceleration of magnitude $\alpha$ in the observer's instantaneous frame of reference. 

How should the astronaut fire the engines to travel on the worldline that maximalises proper time among all possible worldlines under these constraints, and how much proper time $\tau$ do they have?
\end{problem}

The quantities $e_0$ and $L_0$ will be further elucidated in chapter \ref{sec:geometry}; the definitions follow standards found in most relativity texts, such as in \cite{wald2010general, carroll2004spacetime, hartle2003gravity}, and are linked with the timelike and spacelike Killing vector fields. They are conserved along geodesics.

\subsection{Previous work in scientific literature}
The Lewis and Kwan paper \cite{lewis2007no} uses the Eddington-Finkelstein coordinates to show complete radial geodesics from rest can be improved by applying thrust. It extensively refers to Rindler's 1960 paper \cite{rindler1960hyperbolic} to perform calculations of accelerated motion.

A correct qualitative description of how to achieve maximal survival time under the event horizon, with no ambiguity that other types of infall result in longer times can be found in the book \cite{novikov1995black}.

Various papers from Toporensky and affiliated researchers Zaslavskii and Radosz deal with various interesting questions one can ask in this regime. The paper \cite{toporensky2020strategies} gives a first principle explanation using relative flow velocities to give a qualitative answer to the best strategies to maximise either survival time or how much of the outside universe can `contact' the infalling astronaut, establishing that in some situations these goals are not congruent with each other and require different strategies. The paper \cite{radosz2019inside} shows the peculiar velocity of a uniformly accelerated massive particle under the event horizon, as well as how the exchange of electromagnetic signals operates in that regime. The paper \cite{augousti2018speed} also touches on similar issues.

The paper \cite{toporensky2021flow} is a thorough description of geodesic motion, peculiar velocities, and redshifts inside the event horizon of a black hole. The paper \cite{toporensky2023delay} presents a pedagogical example showing momentary impulses can increase survival time, and time interval in the Lema{\^\i}tre frame. It is designed to be suitable as a pedagogical tool.

In the papers \cite{kostic2012analytical} and \cite{cieslik2022revisiting} a discussion of exact analytical expressions for infalling massive and massless particles undergoing geodesic motion is found.

Finally, the article \cite{radosz2023particle} establishes many results regarding particle collisions and kinematics under the event horizon. It also shows surprising links with cosmology and generalizes the Lema{\^\i}tre frame. It comments on issues regarding the stability around the singularity. Also, it establishes similar results around different singularities.

The paper \cite{PhysRevD.52.5832} describes the upper bounds for survival inside black holes in general terms.

The novel results in this paper will be the precise calculation of gains depending on initial parameters; the analytic expression for the case of radial infall; a principle that explains the optimal way to use an engine of finite power that has some application in other contexts; and a discussion of the viability of extending life under the horizon for astrophysical objects and physiologically reachable accelerations.

\subsection{Notes about the calculations performed for this paper}

Almost all calculations done in this text have been confirmed by direct numerical calculation using the equation of motion \eqref{eq:gen_eom}. The calculations were performed in Mathematica using the \textit{xAct} suite \cite{xact}. The notebooks used for the calculations and generation of the plots can be found in a public repository \cite{notebookcode} maintained by the author of this paper. The diagram \ref{fig:carter_penrose} was drawn with a modified version of the code in the paper \cite{olz2013conformal}.

\section{\label{sec:geometry}Analysis of the Schwarzschild spacetime in the region under the event horizon}
In this section we will analyse the Schwarzschild spacetime under the event horizon. To facilitate this analysis, we shall also briefly look into geodesic motion of massive and massless particles.
\subsection{The Schwarzschild spacetime and its' causal structure}
First off, we have to decide on the best coordinate chart for the problem at hand. In the paper \cite{lewis2007no}, the Eddington-Finkelstein coordinates in their original form \eqref{eq:ef_line_element} are used. Their main advantage is that they provide a chart that penetrates the event horizon at $r = 2M$ cleanly. In addition, far away from the Schwarzschild black hole the coordinates are clearly just the Minkowski metric $\eta_{\mu\nu}$. However, the presence of the cross-term $drdt$ means that we do not have a clear separation of the timelike and spacelike coordinates of the metric. We also do not obtain a full analytic continuation. They are a natural choice when considering geodesics that cross the horizon.

When we want to more closely analyse the causal structure of the spacetime, the Kruskal-Szekeres coordinates \eqref{eq:ks_line_element} provide a natural framework for this. Not only do they describe the maximal analytic continuation in one chart, the coordinates are well suited for drawing light cones. Nevertheless, they have the obvious drawback of being much more complicated to perform calculations in, and in many ways it can be difficult to understand exactly what's happening on a kinematical level. We will therefore avoid them -- however, the coordinates based on them will be of use when we analyse the causal structure on a Carter-Penrose diagram.

For completeness, we can use any of the horizon-penetrating coordinates to `pass' from above to below the horizon in an infinitesimally small region around the event horizon.

Since Schwarzschild coordinates \eqref{eq:schw_line_element} are perfectly valid everywhere except for the event horizon, curves described with them can asymptotically reach the event horizon both from above and below the horizon. Looking at the problem statement \ref{problem}, we can see that we're only interested in the spacetime region from $2M \geq r > 0$. The only points at which we cannot use Schwarzschild coordinates are on the hypersurface $2M$. But, any massive test particle's curve can only fall towards the singularity starting from that point; therefore, any curve $\gamma$ describing the path of a massive test particle only has a single point on the surface $2M$. Removing this measure zero set from a continuous curve will not impact the results. We can safely use standard Schwarzschild coordinates, and perform integration in them up to $2M$. They are a natural choice of coordinates for the problem. It is however worth noting that the hypersurface $2M$ is degenerate in this chart, which corresponds to the possibility of infall from two possible regions of the maximally extended spacetime; this will turn out to be intimately related with the sign of the constant of motion representing kinetic energy per unit mass.

Since we are in the interval $r \in (0, 2\pi)$ and $r$ is the timelike coordinate, it is natural for us to switch to using the metric in the form
\begin{equation}\label{eq:schw_line_element_under_horizon}
    ds^2 =  - \left(\frac{2M}{r} - 1\right)^{-1} dr^2 + \left(\frac{2M}{r} - 1\right) dt^2 + r^2 d\Omega^2,
\end{equation}
with coordinates other than $r$ having the range $0 < r < 2M$ and the other coordinates the same as in \eqref{eq:schw_line_element}. Underneath the horizon in these coordinates, $-r$ is the time function (fulfills the requirement that $\nabla f$ is timelike past pointing) \cite{hawking1969existence}. From now on we will reorder coordinates to the form $(r, t, \theta, \phi)$ so that they conform to the metric signature $(-,+,+,+)$. 

We will be describing a curve parametrized by proper time $\tau$: $x^\mu(\tau) = x^\mu = (r(\tau), t(\tau),\theta(\tau),\phi(\tau))$, with associated 4-velocity $u^{\mu}(\tau) = u^{\mu} = \frac{dx^{\mu}}{d\tau}$. 

In Schwarzschild coordinates, the spacetime structure can be read from the line element. We have a total of $4$ Killing vector fields: $\partial_t$ due to staticity, and the three-element rotation group on the $2-$sphere $SO(3)$ associated with the spherical symmetry. We will use spherical symmetry to position our curves in the equatorial plane $\theta(\tau) = \frac{\pi}{2}$ in accordance with the standard methods of simplifying the analysis of geodesics that can be found eg.\ in the references \cite{carroll2004spacetime, wald2010general, cieslik2022revisiting, kostic2012analytical}. Some commentary will be necessary when we get to accelerated motion, but for now this is a good simplifying assumption. The remaining two Killing fields are $\partial_t$ and $\partial_\phi$. Associated with them we have two constants of geodesic motion
\begin{equation}\label{eq:energy_per_unit_mass}
    - u^{\mu}{(\partial_t)}_{\mu} = \left(\frac{2M}{r} - 1\right)^{-1} u^t = e,
\end{equation}
and
\begin{equation}\label{eq:angular_momentum_per_unit_mass}
    u^{\mu}{(\partial_\phi)}_{\mu} = \frac{u^\phi}{r^2} = L.
\end{equation}

The quantity $L$ in equation \eqref{eq:angular_momentum_per_unit_mass} has the interpretation of angular momentum per unit mass. It is the degree of departure from a radial geodesic. Particles with sufficiently high values of $L$ are very unlikely to end up falling into the event horizon, as they will tend to orbit the black hole or escape hyperbolically to infinity \cite{carroll2004spacetime, wald2010general}. This quantity is manifestly covariant, and must be preserved under geodesic motion in the entire region $r>0$ of the manifold, including while crossing the horizon.

The quantity $e$ in equation \eqref{eq:energy_per_unit_mass} has the simplest interpretation when radial infall is considered. When our particle falls starting from rest at infinity, it has the energy per unit mass $e = 1$. If one were to start from `rest' at the event horizon $2M$, one would obtain $e = 0$. We could equally define it as being measured in spatial infinity. It is also manifestly covariant. It appears that in the Schwarzschild spacetime, we cannot sensibly define a $e < 0$, since one could just perform the transformation $t' = -t$ and `flip' the sign without changing the physical interpretation. The sign of $e$ is physically meaningful for a worldline, and negative values correspond to worldlines falling in from the second asymptotically flat region in the analytic continuation of the Schwarzschild spacetime. Imagine you are falling from rest at infinity, so $e = 1$ and $L = 0$. Let's assume you are falling feet down -- prior to crossing the horizon your feet are closer to the black hole, while your head is farther away.

In Eddington-Finkelstein coordinates, the quantity $e$ is equal to $e = g_{tt} u^t + g_{tr} u^r$ \cite{lewis2007no}, and is conserved in the region $r > 0$, including during the crossing of the event horizon. In Schwarzschild coordinates just above the horizon, our 4-velocity $u^{\mu}$ has a timelike component and a 3-velocity directed purely in the $\partial_r$ direction, by \eqref{eq:energy_per_unit_mass} and \eqref{eq:angular_momentum_per_unit_mass}. Since the quantity $L$ is a constant of motion across the entire spacetime and must remain $0$, the only spatial dimension after crossing the horizon in which there can exist displacement between your feet and your head is the one associated with the Schwarzschild coordinate $t$.

By the above argument, under the horizon our 3-velocity must be directed in the $\partial_t$ direction. Looking at the Carter-Penrose diagram \ref{fig:carter_penrose}, the $t$ coordinate in the region II is a spatial coordinate going from $+ \infty$ on the left to $- \infty$ on the right. Falling in from the right gives a decreasing $e$ and from the left an increasing $e$. Therefore, when we're considering a curve through this spacetime, there is a definite meaning to the $\partial_t$ coordinate's direction: a negative $e$ corresponds to a worldline of an observer falling in from the second asymptotically flat region in the analytic continuation.

Under accelerated motion, the values of $e$ and $L$ are no longer restricted to be constant. Using engines, we can change the geodesic we are on to one with different constants of motion. We can aim our engines parallel to the $\partial_t$ direction to change the value of $e$. The direction we should aim at has a definite meaning regardless of the asymptotically flat region we fall in from: aim the engines towards our feet to reduce the absolute value of $e$, and aim towards our head to increase it.

\subsection{Geodesic motion and survival time}

Geodesics in Schwarzschild spacetime that penetrate the horizon have been analysed in multiple sources, (for example, the papers \cite{cieslik2022revisiting, kostic2012analytical}), so we will not reproduce their findings here. Our analysis of accelerated motion will include geodesics as a special subcase, for $a^\mu = 0$.

\begin{figure}[t]
    \includegraphics[width=0.7\linewidth]{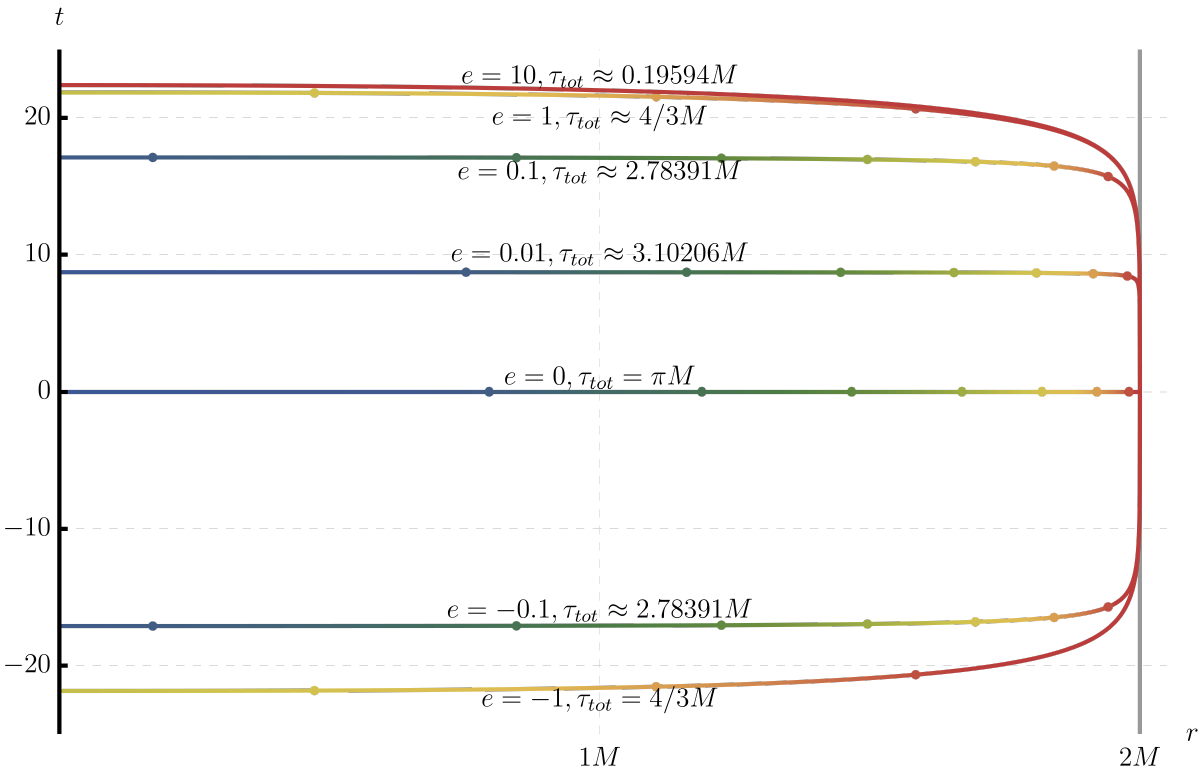}
    \caption{The geodesics of the radial infall $L_0 = 0$ of massive particles into a black hole of arbitrary mass $M$ in geometric units, starting $5*10^{-5}M$ under the event horizon. $t$ is chosen to start at $0$. Colors of the individual curves go from red at early proper times to blue at late proper times; the individual dots on each geodesics signify $\frac{\pi M}{8}$ intervals as measured on the clock carried by the falling observer. Clearly visible is the different total proper time $\tau_{tot}$ spent under the event horizon before hitting the singularity, depending on the initial constant $e_0$. Coordinates $r$ (x axis) and $t$ (y axis) are the Schwarzschild chart \eqref{eq:schw_line_element}. The direct geodesic \eqref{eq:maximal_geodesic_time} with $e_0 = 0$ survives the longest.}
    \label{fig:geodesics:e}
\end{figure}      

Of great interest to us, however, is the length of a geodesic. The length of the geodesic integrated across the monotonically decreasing coordinate $r$ is \cite{hartle2003gravity}
\begin{equation}\label{eq:geodesic_proper_time}
    \tau = \int_{0}^{2M} dr \left[ e^2 +\left( 1 + \frac{L^2}{r^2} \right) \left(\frac{2M}{r} - 1\right) \right]^{-\frac{1}{2}}.
\end{equation}
This time is maximised when $e = L = 0$ for all $r < 2M$, and is equal to
\begin{equation}\label{eq:maximal_geodesic_time}
    \tau_{\textrm{opt}} = \int_{0}^{2M} dr \left[ \frac{2M}{r} - 1 \right]^{-\frac{1}{2}} = M \pi.
\end{equation}

This geodesic is what is often calculated as the upper bound of survival time \cite{hartle2003gravity, carroll2004spacetime}, but we note that it is a specific geodesic that \textbf{no observer falling from above the event horizon will actually be on}: by definition, at least $e$ will be non-zero in this case. Only through carefully thrusting with nearly all our might at the event horizon can we asymptotically approach this geodesic. However, it is the maximal time an observer can survive in the interior of the event horizon. We will refer to this geodesic as the `optimal geodesic' or `maximal geodesic', and to observers on such a geodesic as `optimal inertial observers` in the rest of the article.

Figure \ref{fig:geodesics:e} shows infall that's purely radial with various $e$; figure \ref{fig:geodesics:l} shows infall that starts as equatorial with various $L$ and $e = 0$. We can see the most direct path towards the singularity is the longest living one, and with higher values of $e$ and $L$ we can lower the survival time at our leisure.

Before we get into analysing this in more depth, let's assume we possess an engine of arbitrary power. It is quite obvious that if we have access to that, the correct strategy to maximise lifetime is to momentarily pulse the engine to kill the angular momentum and the kinetic energy -- to arrive at the maximal geodesic. This thrust should work against the direction of our 3-velocity in relation to the frame of a hypothetical optimally falling observer \eqref{eq:maximal_geodesic_time}. Since the maximal geodesic has the same characteristic regardless of current $r$ coordinate, we can further surmise that the maximal geodesic is the goal we're trying to reach at all times $r$. This argument has been elaborated on more in papers \cite{toporensky2020strategies, toporensky2023delay}.

We can also look at the symmetry of $e$ in regards to sign again: if we thrust to lower ourselves to the maximal geodesic and then keep thrusting, we will actually build up the absolute value of $e$ again. This way, we revert to a symmetric situation of infall from the second asymptotically flat region in the analytic continuation. This means that at some point, we should turn the engine off to avoid overshooting the maximal geodesic.

\begin{figure}[b]
    \includegraphics[width=0.7\linewidth]{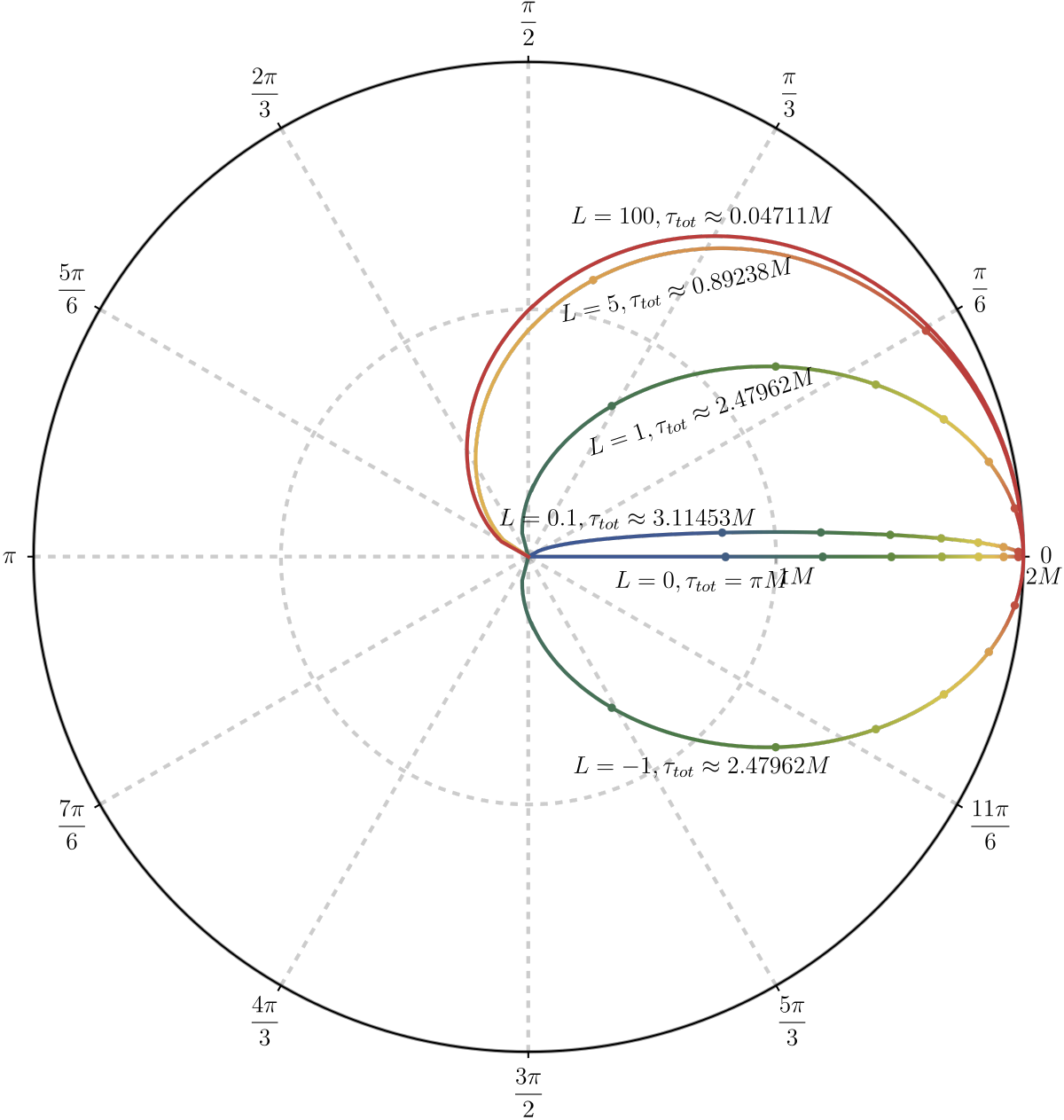}
    \caption{The geodesics of the infall with $e_0 = 0$ of massive particles into a black hole of arbitrary mass $M$ in geometric units, starting $5*10^{-5}M$ under the event horizon. Colors of the individual curves go from red at early proper times to blue at late proper times; the individual dots on each geodesics signify $\frac{\pi M}{8}$ intervals as measured on the clock carried by the falling observer. Clearly visible is the different total proper time $\tau_{tot}$ spent under the event horizon before hitting the singularity, depending on the initial constant $L_0$. Coordinates $r$ and $\phi$ are the Schwarzschild chart \eqref{eq:schw_line_element}. The direct geodesic \eqref{eq:maximal_geodesic_time} with $L_0 = 0$ is the longest.}
    \label{fig:geodesics:l}
\end{figure}      

\section{\label{sec:motion}Analysis of accelerated motion under the event horizon}
The idea of a momentary impulse engine has allowed us to show that using engines is a viable strategy, and the maximal geodesic is the target we want to reach. However, we have no quantitative information about our expected lifetime based on parameters, nor is such an engine a physically viable thing. In this section, we will strive to describe accelerated motion in full, beginning with special cases and moving towards the general case.

Under accelerated motion, particles obey the equation of motion
\begin{equation}\label{eq:gen_eom}
    a^{\mu} = u^{\alpha} u^{\mu}_{;\alpha} = \frac{du^{\mu}}{d\tau} + \Gamma^{\mu}_{\alpha\beta} u^{\alpha} u^{\beta}
\end{equation}
where the semicolon represents the covariant derivative, $\Gamma^{\mu}_{\alpha\beta}$ are the Christoffel symbols, and $a^{\mu}$ is the 4-acceleration. As we can see, 4-acceleration represents the degree of departure from simple geodesic motion. Because of that, it is constrained by orthogonality with the 4-velocity 
\begin{equation}\label{eq:a_orthogonal}
    a^{\alpha} u_{\alpha} = 0,
\end{equation}
and the 4-acceleration is normalised such that its' square
\begin{equation}\label{eq:a_norm}
    a^{\alpha} a_{\alpha} = \alpha^2
\end{equation}
is the squared magnitude of the acceleration felt by the observer in their instantaneous inertial frame.

We also have the normalization of the 4-velocity
\begin{equation}\label{eq:u_norm}
    u^{\alpha} u_{\alpha} = -1.
\end{equation}

We use two of the three Killing vectors on the 2-sphere to analyse motion in the equatorial plane $\theta(\tau) = \frac{\pi}{2}$. While nothing stops us from actually deviating from the equator using our engine, deflecting away from the equator will always worsen our situation: we'll spend valuable engine power changing the inclination of our trajectory instead of arresting excess momentum keeping us off the optimal path. We shall stick to the equator.

This means at every event along a worldline we have 6 variables (three components of 4-velocity, three components of 4-acceleration) and only three constraints \eqref{eq:u_norm}, \eqref{eq:a_norm}, \eqref{eq:a_orthogonal} -- since the constants of motion $e(\tau)$ and $L(\tau)$ are no longer unchanging. Along with the two variables defining initial conditions $e_0 = e(0)$ and $L_0 = L(0)$, that leaves us a degree of freedom -- the way we are going to control our spacecraft. 

\subsection{Frame field of optimal fall}\label{section:framefield}
Our further arguments will make extensive use of a particular frame field. As shown in \ref{fig:geodesics:e} and \ref{fig:geodesics:l}, the geodesic which achieves the maximal survival time \eqref{eq:maximal_geodesic_time} must have $L = e = 0$. An observer with that geodesic at $t = 0, \phi = 0, \theta = \frac{\pi}{2}$ has the 4-velocity
\begin{equation}\label{eq:u_opt}
    u_{\textrm{opt}}^\mu = (- \sqrt{\frac{2M}{r} - 1},0,0,0).
\end{equation}

Since this observer is in free fall towards the singularity, at every point of their journey they carry an inertial frame aligned with the vectors $\partial_t, \partial_\phi$ and $\partial_\theta$. It is in this frame that the measurement of the values of $L$ and $e$ of an observer on a worldline crossing the frame can be performed.

Under the horizon, the Killing vector field $\partial_t$ becomes a spacelike Killing field. $\partial_\phi$ is also a spacelike Killing field, and the 2-sphere has two further Killing vectors which allow us to arbitrarily rotate our coordinate chart. Using those fields on our optimal observer, we can translate his position in the coordinate $t$ and rotate freely in the coordinates $\partial_\phi$ and $\partial_\theta$ using passive transformations. In this way, we can span the entire manifold under the horizon by a frame field of optimal observers towards the singularity. Thus, using optimal observers, we obtain a well defined frame at every point in the region $r < 2M$.

\subsection{Purely radial infall}
The situation which has been most studied in literature is the case of radial infall. In this situation, we can without loss of generality assume the coordinate $\phi = 0$ across the entire worldline. This means $L_0 = 0$ and we are already on the maximal geodesic with $u^\phi = 0$. We don't need to apply any acceleration along $\partial_\phi$, so $a^\phi = 0$. Eliminating these variables means we have 4 variables for our three constraints and initial condition $e_0 = e(0)$, meaning the system has an explicit solution.

The equations of motion for this system are
\begin{equation}\label{eq:eom_radial}
    \begin{split}   
        - t'' & = \frac{2M r' t'}{r (-2M + r)} + \frac{\alpha r^{3/2} r'}{\sqrt{(2M - r) \left[ r^2 r'^2 - (2M - r)^2 t'^2 \right]}} \\
        - r'' & = - \frac{M r'^2}{r (-2M + r)} + \frac{M(-2M + r) t'^2}{r^3} + \frac{\alpha (-2M + r) t'}{\sqrt{r(2M - r) \left[ r^2 r'^2 - (2M - r)^2 t'^2 \right]}}
    \end{split}
\end{equation}
where the argument $\tau$ of the parametrized coordinates has been suppressed for brevity. The $'$ represents differentiation over the variable $\tau$. This system of equations can be integrated numerically, and we have used it to verify further results. The initial conditions can be set arbitrarily close to $r(0) = 2M$ at $t(0) = 0$, with $t'(0)$ calculated from \eqref{eq:energy_per_unit_mass} and $r'(0)$ from \eqref{eq:u_norm}.

\begin{figure}[t]
    \includegraphics[width=0.7\linewidth]{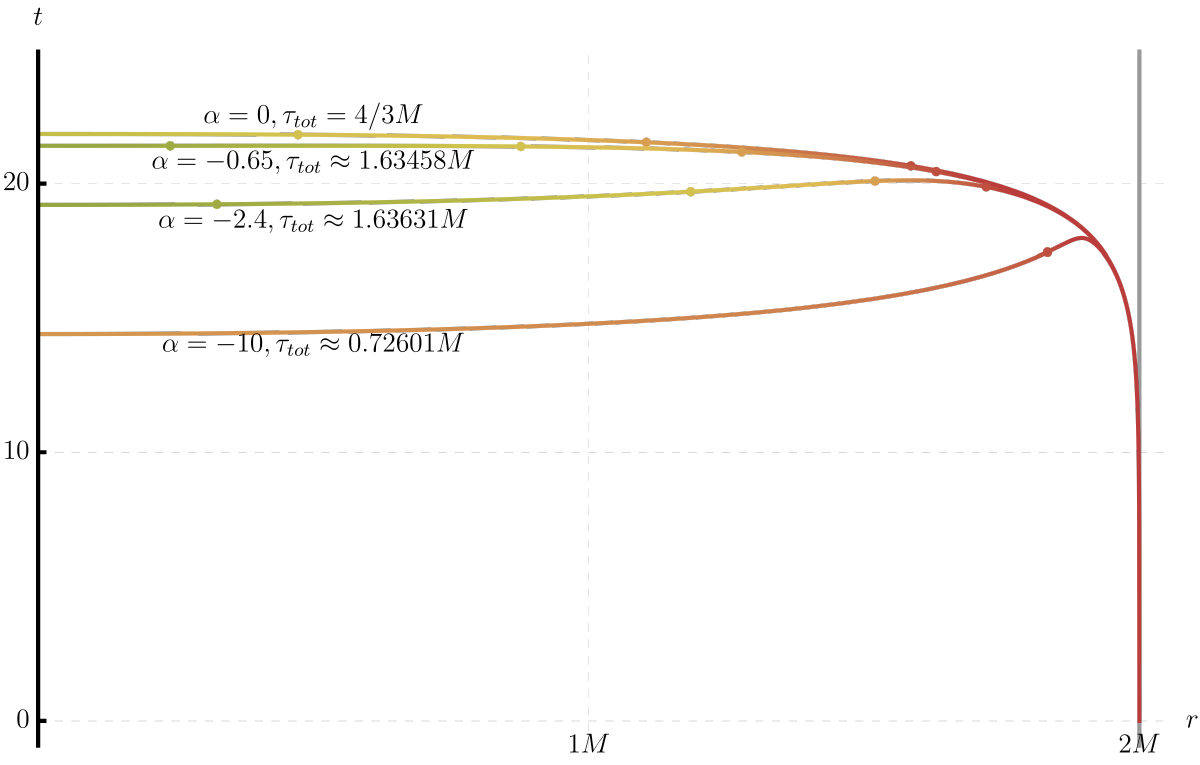}
    \caption{Curves of accelerated motion after infall with $L_0 = 0, e_0 = 1$ for various values of acceleration $a$. Excessive accelerations result in overshooting the perfect geodesic and thus in lowering the time survived; however, some values of $a$ extend the observer's life.}
    \label{fig:curves:e}
\end{figure}      

While the value $e(\tau)$ is no longer constant, in every instantenous frame associated with the frame field of optimal infall it retains its definition \eqref{eq:energy_per_unit_mass}. Taking the absolute derivative of \eqref{eq:energy_per_unit_mass} written as a tensor equation with respect to proper time gives
\begin{equation}
    \frac{De}{d\tau} = \frac{D}{d\tau} (u_\mu (\partial_t)^\mu) = \frac{D u_\mu}{d\tau} (\partial_t)^\mu + \frac{D (\partial_t)^\mu}{d\tau} u_\mu.
\end{equation}
$e$ on the left hand side is a scalar, so the left hand side is the normal derivative; the second term on the right hand side is $0$ by the constancy of $\partial_t$, while the absolute derivative in the first term on the right hand side written out using components is simply $g_{tt} a^t$, from equation \eqref{eq:gen_eom}. This gives
\begin{equation}
    \frac{de}{d\tau} = g_{tt} a^t.
\end{equation}
The left hand side can be expanded with 
\begin{equation}\label{eq:dedr_lhs}
    \frac{de}{d\tau} = \frac{de}{dr} \frac{dr}{d\tau} = \frac{de}{dr} u^r.
\end{equation}
Meanwhile, for the right hand side, we use \eqref{eq:a_norm}, \eqref{eq:u_norm} and \eqref{eq:a_orthogonal} to express $a^t$ in terms of $\alpha$ and $u^r$ (choosing sign so that it matches our choice of the definition of $e$ being lowered by negative $\alpha$):
\begin{equation}\label{eq:dedr_rhs}
    a^t = - \alpha \sqrt{\frac{r}{2M - r} + (u^t)^2} = - \alpha \sqrt{g_{tt}^{-1}} \sqrt{1 + (u^t)^2 g_{tt}} = - \alpha g_{tt}^{-1} u^r.
\end{equation}

Combining \eqref{eq:dedr_lhs} and \eqref{eq:dedr_rhs} all together, we get
\begin{equation}\label{eq:dedr}
    \frac{de}{dr} = - \alpha,
\end{equation}
where all metric components have cancelled each other out due to the fact $g_{tt} = g_{rr}^{-1}$.

The equation \eqref{eq:dedr} can be integrated with initial value $e(2M) = e_0$ to obtain the linear relation between acceleration applied in the $\partial_t$ spatial direction and the value (now in terms of the $r$ coordinate)
\begin{equation}\label{eq:e_of_r}
    e(r) = - \alpha r + 2M \alpha + e_0.
\end{equation}

Not only does this confirm the findings of the paper \cite{lewis2007no}, it allows us to use equation \eqref{eq:geodesic_proper_time} in a new context. Replacing $e$ with $e(r)$, we get
\begin{equation}\label{eq:tau_radial}
    \tau = \int_{0}^{2M} dr \left[ \left( -\alpha r + 2M + e_0 \right)^2 - \left(1 - \frac{2M}{r} \right) \right]^{-\frac{1}{2}},
\end{equation}
where the right side is wholly in terms of $r$. We can perform a substitution $u \to r^{-1}$ to obtain
\begin{equation}\label{eq:tau_r_sub1}
    \tau = \frac{1}{\sqrt{2M}} \int^{\infty}_{\frac{1}{2M}} \frac{du} {u \sqrt{u^3 + (2 \alpha^2 M + 2 \alpha e_0 + \frac{{e_0}^2}{2M} - \frac{1}{2M}) u^2 - (2 \alpha^2 + \alpha \frac{e_0}{2M}) u + \frac{\alpha^2}{2M} }}.
\end{equation}

This is an elliptic integral \cite{byrd2013handbook}. We will represent it in the Carlson symmetric elliptic integral form \cite{carlson1977special,NIST:DLMF}, since those have the advantage of being symmetric and thus easier in analysis, contain fewer branch cuts and behave better computationally \cite{Carlson_1995}.

The function under the square root in the denominator of \eqref{eq:tau_r_sub1} has the following roots:
\begin{equation}\label{eq:tau_r_roots}
    \begin{split}
        u_1 = &\frac{\sqrt[3]{\sqrt{\left(\frac{27 a^2}{2 M}-9 C D+2 D^3\right)^2+4 \left(3 C-D^2\right)^3}-\frac{27 \alpha^2}{2 M}+9 C D-2 D^3}}{3 \sqrt[3]{2}}\\
        + &\frac{\sqrt[3]{2} \left(D^2-3 C\right)}{3 \sqrt[3]{\sqrt{\left(\frac{27 \alpha^2}{2 M}-9 C D+2 D^3\right)^2+4 \left(3 C-D^2\right)^3}-\frac{27 \alpha^2}{2 M}+9 C D-2 D^3}}-\frac{D}{3} \\
        u_2 = &\frac{\left(-1+i \sqrt{3}\right) \sqrt[3]{\sqrt{\left(\frac{27 \alpha^2}{2 M}-9 C D+2 D^3\right)^2+4 \left(3 C-D^2\right)^3}-\frac{27 \alpha^2}{2 M}+9 C D-2 D^3}}{6 \sqrt[3]{2}}\\ 
        + &\frac{\left(1+i \sqrt{3}\right) \left(3 C-D^2\right)}{3\ 2^{2/3} \sqrt[3]{\sqrt{\left(\frac{27 \alpha^2}{2 M}-9 C D+2 D^3\right)^2+4 \left(3 C-D^2\right)^3}-\frac{27 \alpha^2}{2 M}+9 C D-2 D^3}}-\frac{D}{3} \\
        u_3 = &-\frac{\left(1+i \sqrt{3}\right) \sqrt[3]{\sqrt{\left(\frac{27 \alpha^2}{2 M}-9 C D+2 D^3\right)^2+4 \left(3 C-D^2\right)^3}-\frac{27 \alpha^2}{2 M}+9 C D-2 D^3}}{6 \sqrt[3]{2}} \\ 
        + &\frac{\left(1-i \sqrt{3}\right) \left(3 C-D^2\right)}{3\ 2^{2/3} \sqrt[3]{\sqrt{\left(\frac{27 \alpha^2}{2 M}-9 C D+2 D^3\right)^2+4 \left(3 C-D^2\right)^3}-\frac{27 \alpha^2}{2 M}+9 C D-2 D^3}}-\frac{D}{3},
    \end{split}
\end{equation}
where $C = (2 \alpha^2 + \alpha \frac{e_0}{2M})$ and $D = (2 \alpha^2 M + 2 \alpha e_0 + \frac{{e_0}^2}{2M} - \frac{1}{2M})$. 

Equation \eqref{eq:tau_r_sub1} can be represented using \eqref{eq:tau_r_roots} as
\begin{equation}\label{eq:tau_r_sub2}
    \tau = \frac{1}{\sqrt{2M}} \int^{\infty}_{\frac{1}{2M}} \frac{du} {u \sqrt{(u - u_1)(u - u_2)(u - u_3)}}.
\end{equation}

To obtain a Carlson symmetric form elliptic integral, we only need to bring the lower bound to $0$ by substitution $u' \to u - \frac{1}{2M}$:
\begin{equation}\label{eq:tau_r_carlsonj}
    \begin{split}
        \tau & = \frac{1}{\sqrt{2M}} \int^{\infty}_{0} \frac{du'} {\left(u' + \frac{1}{2M}\right) \sqrt{\left(u' - u_1 + \frac{1}{2M}\right)\left(u' - u_2 + \frac{1}{2M}\right)\left(u' - u_3 + \frac{1}{2M}\right)}} \\ 
             & = \frac{2}{3} \frac{R_J\left(\frac{1}{2M} - u_1,\frac{1}{2M} - u_2,\frac{1}{2M} - u_3,\frac{1}{2M}\right)}{\sqrt{2M}},
    \end{split}
\end{equation}
which is the analytic expression for the proper time along the accelerating curve from $r = 2M$ to $r = 0$ in infall with no angular momentum. The integral \eqref{eq:tau_r_carlsonj} can be represented in Legendre form integrals \cite{NIST:DLMF}:
\begin{equation}\label{eq:tau_r_legendre}
    \tau = \frac{2}{\alpha_L^2 \sqrt{2M}} \frac{\Pi(\phi_e,\alpha_L^2,k) - F(\phi_e,k)}{(u_3 - x)^\frac{3}{2}},
\end{equation}
where $\alpha_L^2 = \frac{u_3 - \frac{1}{2M}}{u_3 - u_1}$, $k = \sqrt{\frac{u_3-u_2}{u_3-u_1}}$ and $\phi_e = \arccos(\sqrt{\frac{u_1}{u_3}})$. We will stick to the form \eqref{eq:tau_r_carlsonj}, as it is better behaved.

\begin{figure}[t]
    \includegraphics[width=0.7\linewidth]{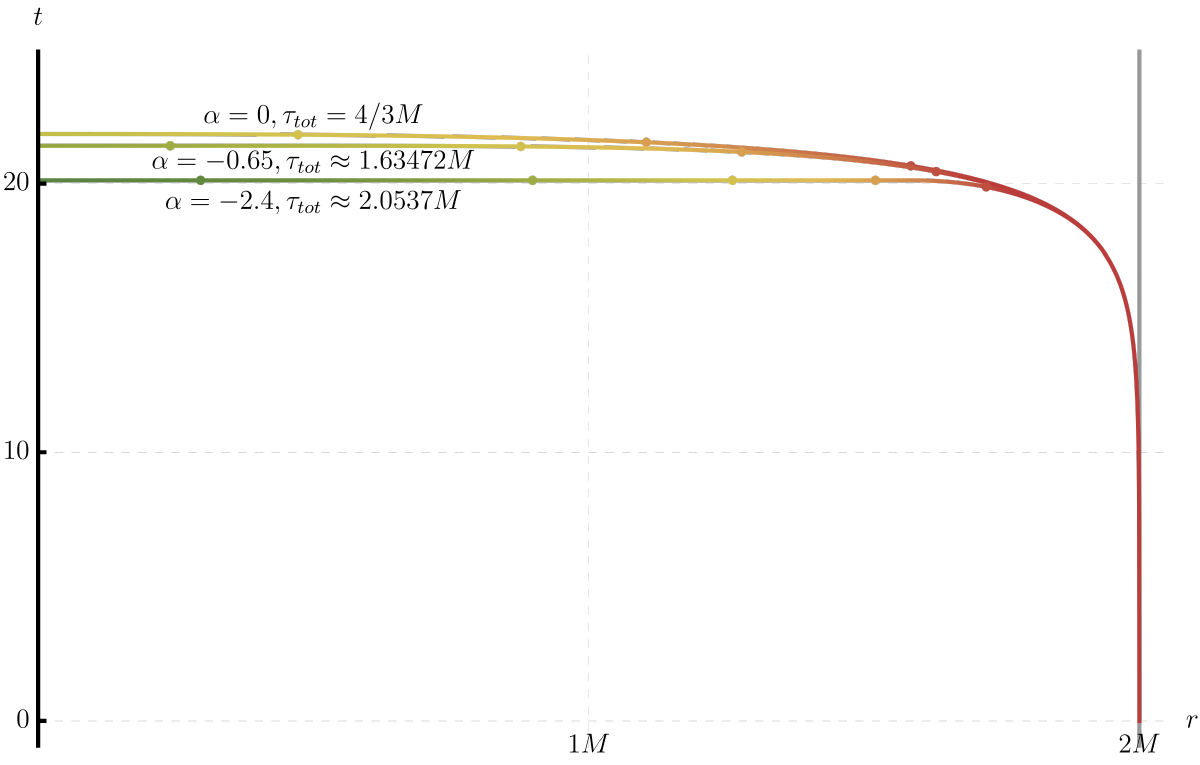}
    \caption{Curves of accelerated motion after infall with $L_0 = 0, e_0 = 1$ for various values of acceleration $\alpha$. If we reach the optimal geodesic, we cease accelerating further, and this way, the more acceleration we have, the more we can live. Compare with figure \ref{fig:curves:e}.}
    \label{fig:curves:e_optimal}
\end{figure}      

One more thing to note: we are interested in reaching the optimal geodesic. Therefore, sometimes we should cease our acceleration before reaching $r = 0$. From the equation \eqref{eq:e_of_r} we can obtain the moment of reaching the optimal geodesic
\begin{equation}
    r_{\text{opt}} = 2M + \frac{e_0}{\alpha},
\end{equation}
and use a substitution $u' \to u - \frac{1}{r_{\text{opt}}}$ on \eqref{eq:tau_r_sub2} with a different lower bound $\frac{1}{r_{\text{opt}}}$ to obtain an integral that we can subtract from integral \eqref{eq:tau_r_carlsonj}:
\begin{equation}\label{eq:tau_r_carlsonj_to_opt}
    \tau(r_{\text{opt}}) = \frac{2}{3} \frac{R_J\left(-x + \frac{1}{2M},-y + \frac{1}{2M},-z + \frac{1}{2M},\frac{1}{2M}\right) - R_J\left(-x + \frac{1}{r_{\text{opt}}},-y + \frac{1}{r_{\text{opt}}},-z + \frac{1}{r_{\text{opt}}}, \frac{1}{r_{\text{opt}}}\right)}{\sqrt{2M}},
\end{equation}
which can then be added to the proper geodesic time from $r_{\text{opt}}$ to $0$ to get the final proper time for movement where we cease firing the engine at the most optimal moment.

\begin{figure}[t]
    \begin{subfigure}{0.49\textwidth}
        \includegraphics[width=0.9\linewidth]{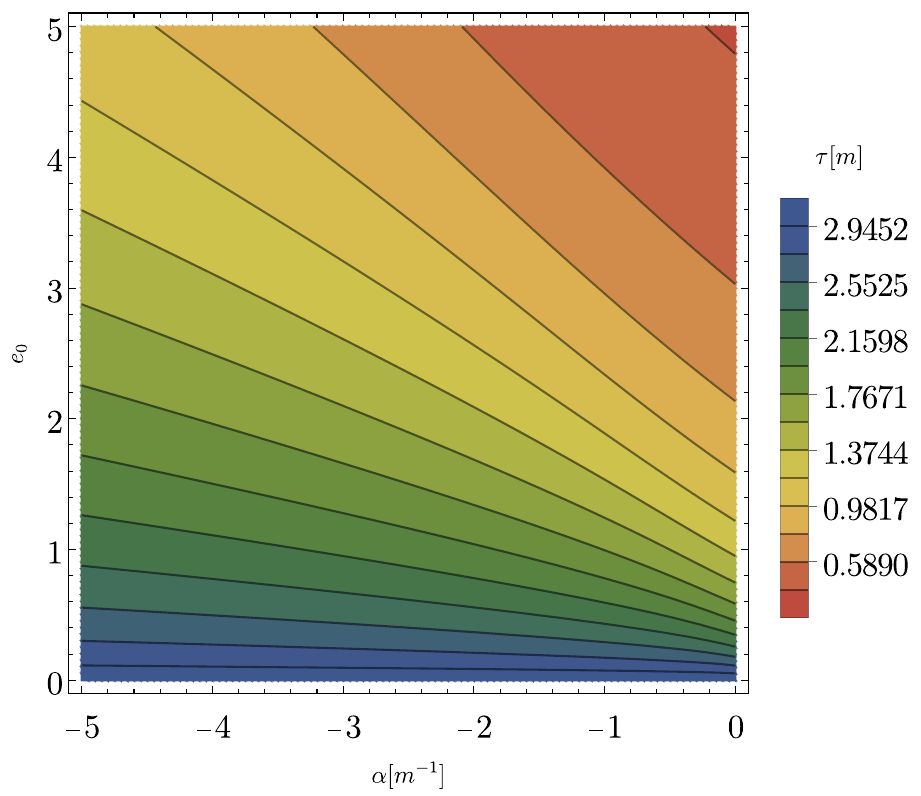} 
        \caption{Contour plot of proper time for radial infall depending on the parameters $\alpha$ and $e_0$. Colors retain meaning of proper time $\tau$ from worldline graphs. Contours are separated by $\frac{\pi}{16}$.}
        \label{fig:plots:e-contour}
    \end{subfigure}
    \begin{subfigure}{0.49\textwidth}
        \includegraphics[width=0.9\linewidth]{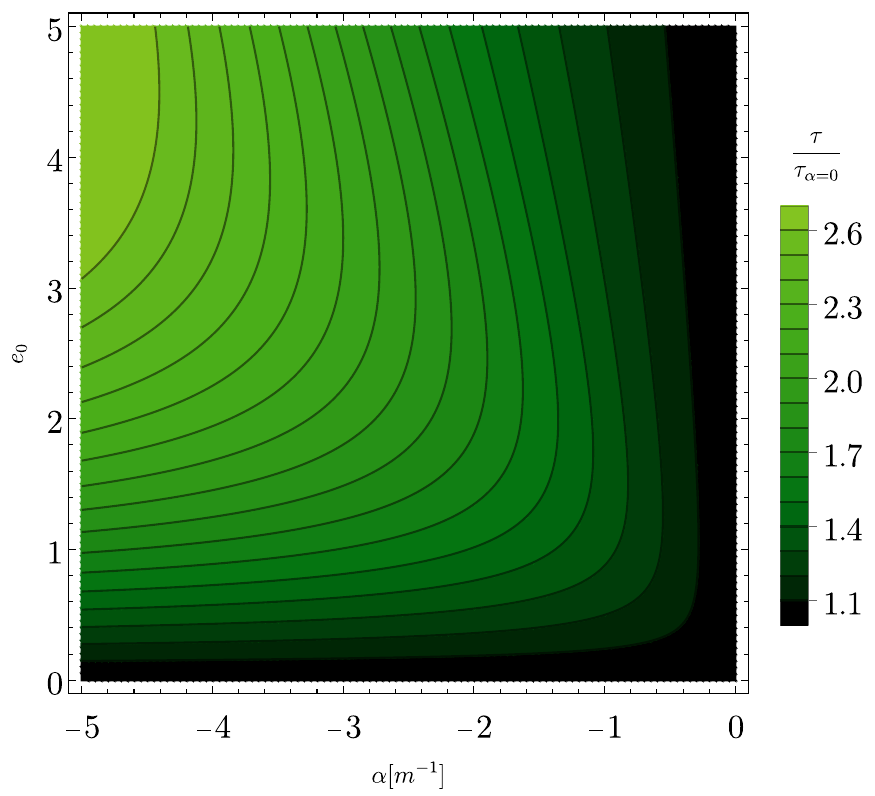} 
        \caption{Contour plot of the fraction $\frac{\tau}{\tau_{\alpha=0}}$ representing the relative gain thanks to using engines.}
        \label{fig:plots:e-contour-gain}
    \end{subfigure}
    \caption{Plots of proper time and time gained for accelerated motion in a black hole with $M = 1 m$.}
    \label{fig:plots:e}
\end{figure}

The equations contain the formulas for geodesic time, too -- simply set $\alpha = 0$. Therefore, the proper time of the motion can be easily calculated using a generic function that computes three Carlson elliptic integrals in a computer algebra program such as Mathematica. Figures \ref{fig:plots:e} show the proper time for a variety of parameters.

Also, from \eqref{eq:tau_radial} by differentiating both sides, and by using \eqref{eq:u_norm} we can obtain an analytic expression for the 4-velocity as a function of $r$:
\begin{equation}
    \begin{split}
        u^{r} & = - \sqrt{ \left( -\alpha r + 2M + e_0 \right)^2 + \left(\frac{2M}{r} - 1\right) } \\
        u^{t} & = \pm \left(\frac{2M}{r} - 1\right)^{-1} \sqrt{ \left(\frac{2M}{r} - 1\right) + (u^r)^2 }.
    \end{split}
\end{equation}

\subsection{Infall with angular momentum}
Another simplifying assumption we can have is to assume we are starting with energy $e_0 = 0$, but with non-zero angular momentum. While no observer falling in from the outside will ever start with these initial conditions on the event horizon, it is instructive to eliminate the impact of the displacement in $\partial_t$.

It must be noted that extreme values of angular momentum $L_0 > \sqrt{12}M$ combined with very low values of $e_0$ are not realistic for actual travellers, especially ones that fall into the black hole on a geodesic, since they would be repelled by the centrifugal potential barrier \cite{wald2010general}. However, in the case of a rocket equipped observer, it can happen that crossing the event horizon is unavoidable, but we can delay it with our engines. In that case, increasing angular momentum can be the correct way to delay hitting the event horizon, so initial parameters above $\sqrt{12}M$ are not as impractical as they appear. Regardless, it is a very useful case to focus on for analysis, which we will now show.

With $t = 0$, we have $e_0 = 0$ and $u^t = 0$. We don't need to apply any acceleration in the spatial $t$ coordinate, and $a^t = 0$. We have 4 variables for our three constraints and initial condition $L_0 = L(0)$.

\begin{equation}\label{eq:eom_axial}
    \begin{split}   
        \phi'' & = - \frac{2 r' \phi'}{r} + \frac{\alpha r'}{\sqrt{r^2 ( r'^2 + r( -2M + r) \phi'^2}} \\
        r'' & = \frac{M r'^2}{r (-2M + r)} - (2M - r) \phi'^2 + \frac{\alpha r (-2M + r) \phi'}{\sqrt{r^2 ( r'^2 + r( -2M + r) \phi'^2)}}
    \end{split}
\end{equation}
are the equations of motion in this configuration. Just as before, we have used a numerical integration of this system for verification of results. Initial conditions are calculated arbitrarily close to the horizon $r(0) = 2M$ at $\phi(0) = 0$, with $\phi'(0)$ calculated from \eqref{eq:angular_momentum_per_unit_mass} and $r'(0)$ from \eqref{eq:u_norm}.

\begin{figure}[b]
    \includegraphics[clip,trim=0 9cm 0 0,width=0.7\linewidth]{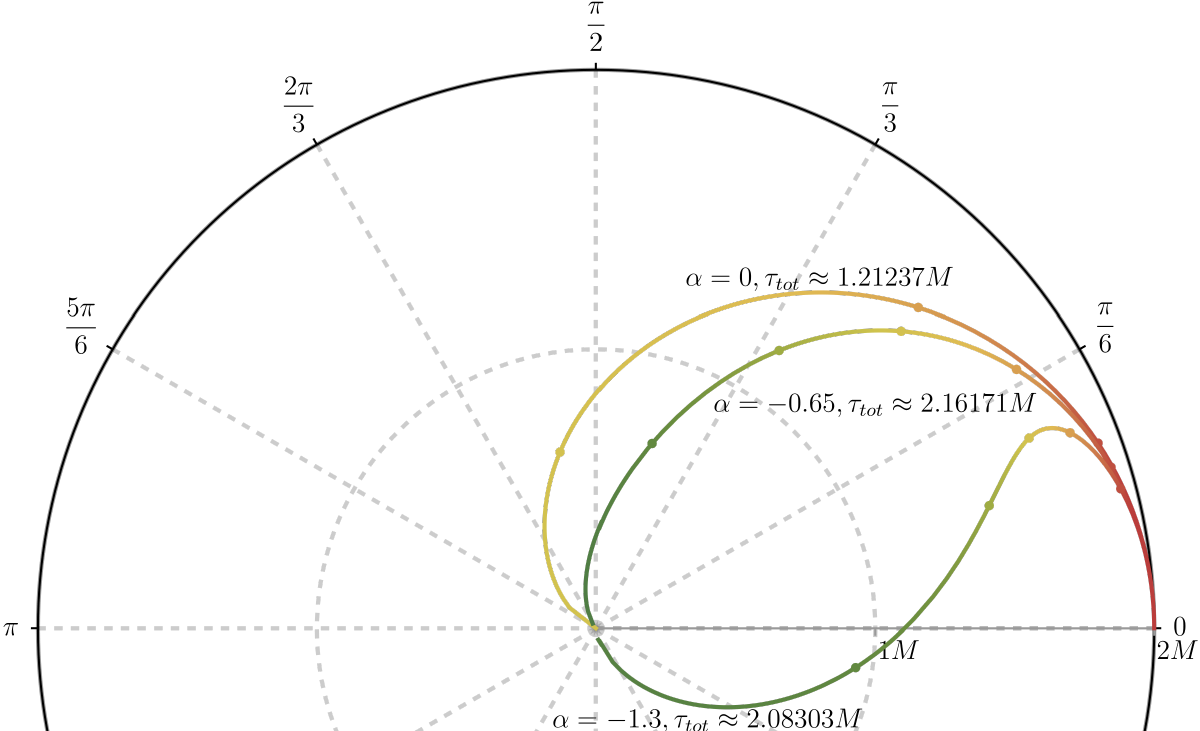}
    \caption{Curves of accelerated motion after infall with $L_0 = 3.5, e_0 = 0$ for various values of acceleration $\alpha$. Excessive accelerations result in overshooting the perfect geodesic and thus in lowering the time survived; however, some values of $a$ extend the observer's life.}
    \label{fig:curves:l}
\end{figure}      

In every instantaneous frame associated with the optimal observer the value $L(\tau)$ (no longer constant) retains its definition \eqref{eq:angular_momentum_per_unit_mass}. Taking the absolute derivative of the tensor form of \eqref{eq:angular_momentum_per_unit_mass} with respect to proper time gives
\begin{equation}
    \frac{DL}{d\tau} = \frac{D}{d\tau} (u_\mu (\partial_\phi)^\mu) = \frac{D u_\mu}{d\tau} (\partial_\phi)^\mu + \frac{D (\partial_\phi)^\mu}{d\tau},
\end{equation}
for which we will use the fact $L$ is a scalar and the constancy of $\partial_\phi$, along with the definition of $a^\phi$ from \eqref{eq:gen_eom} to get
\begin{equation}
    \frac{dL}{d\tau} = g_{\phi\phi} a^\phi.
\end{equation}
The left side can be expanded with 
\begin{equation}\label{eq:dLdr_lhs}
    \frac{dL}{d\tau} = \frac{dL}{dr} \frac{dr}{d\tau} = \frac{dL}{dr} u^r.
\end{equation}
For the right side, we use \eqref{eq:a_norm}, \eqref{eq:u_norm} and \eqref{eq:a_orthogonal} to express $a^\phi$ in terms of $\alpha$ and $u^r$ (choosing sign arbitrarily, since the sign of $L$ has no physical relevance the way the sign of $e$ had):
\begin{equation}\label{eq:dLdr_rhs}
    a^\phi = - \alpha r^{-1} \sqrt{1 + r^2 (u^\phi)^2} = - \alpha r^{-1} \sqrt{g_{tt}^{-1}} u^r.
\end{equation}
Equations \eqref{eq:dLdr_lhs} and \eqref{eq:dLdr_rhs} give 
\begin{equation}\label{eq:dLdr}
    \frac{dL}{dr} = - \frac{\alpha r}{\sqrt{\frac{2M}{r} - 1}}.
\end{equation}
This time, the metric components do not cancel neatly, and the behavior of $L$ under acceleration is radically different.

Equation \eqref{eq:dLdr_rhs} can be integrated with initial value $L(2M) = L_0$ to obtain the relation between acceleration applied in the $\partial_\phi$ spatial direction and the value of $L$ (now in terms of the $r$ coordinate)
\begin{equation}\label{eq:L_of_r}
    L(r) = L_0 + \alpha \left( \frac{1}{2}  r (3M + r) \sqrt{\frac{2M}{r} - 1}+ 3M^2 \arctan\sqrt{\frac{2M}{r} - 1}\right).
\end{equation}

We can insert this into the equation \eqref{eq:geodesic_proper_time} to obtain:
\begin{equation}\label{eq:tau_angular}
    \tau = \int_{0}^{2M} dr \left[ \left( 1 + \frac{\left[ L_0 + \alpha \left( \frac{1}{2} r (3M + r) \sqrt{\frac{2M}{r} - 1} + 3M^2 \arctan\sqrt{\frac{2M}{r} - 1}\right) \right]^2}{r^2}\right) \left(\frac{2M}{r} - 1\right) \right]^{-\frac{1}{2}},
\end{equation}

\begin{figure}[t]
    \begin{subfigure}{0.49\textwidth}
        \includegraphics[width=0.9\linewidth]{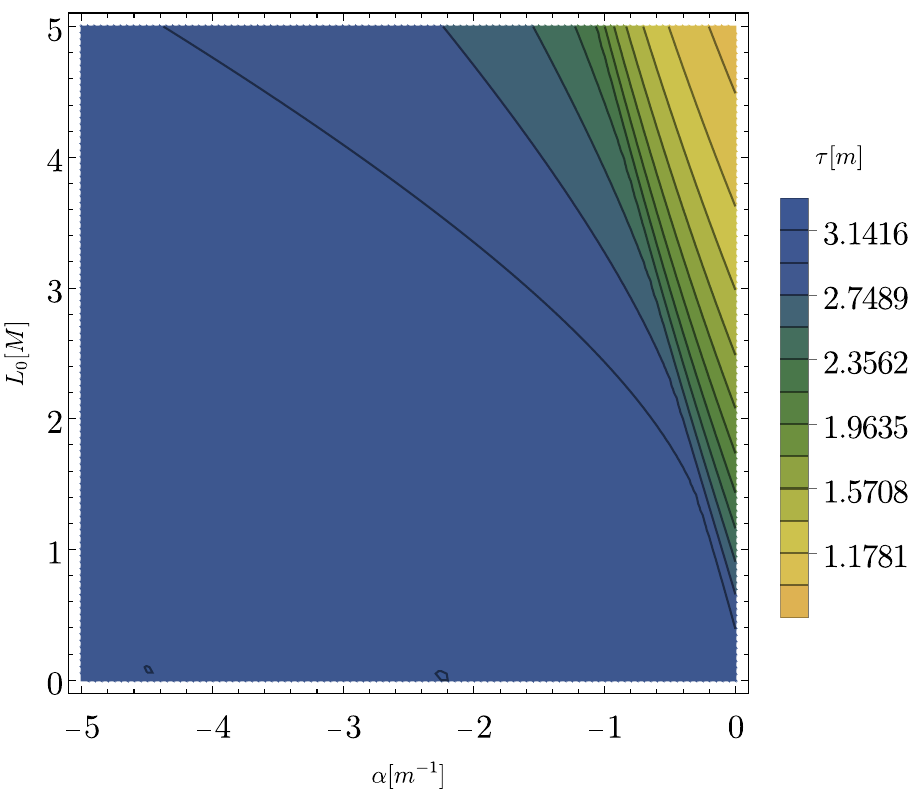} 
        \caption{Contour plot of proper time for longitudinal infall depending on the parameters $\alpha$ and $L_0$. Colors retain meaning of proper time $\tau$ from worldline graphs. Contours are separated by $\frac{\pi}{16}$.}
        \label{fig:plots:l-contour}
    \end{subfigure}
    \begin{subfigure}{0.49\textwidth}
        \includegraphics[width=0.9\linewidth]{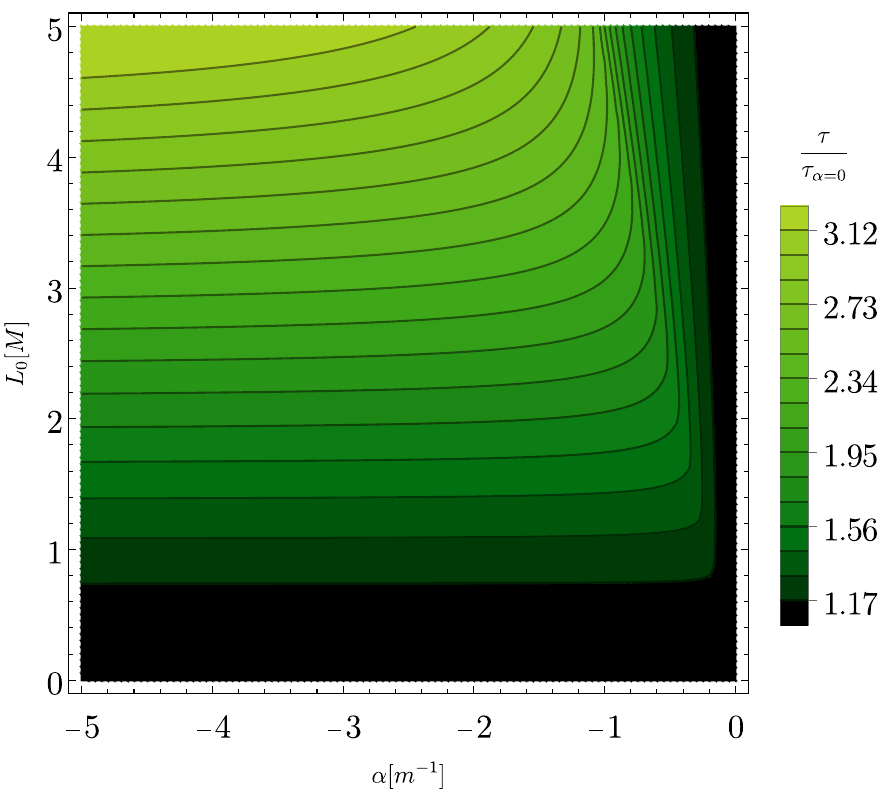} 
        \caption{Contour plot of the fraction $\frac{\tau}{\tau_{\alpha=0}}$ representing the relative gain thanks to using engines.}
        \label{fig:plots:l-contourgain}
    \end{subfigure}
    \caption{Plots of proper time and time gained for accelerated motion in a black hole with $M = 1 m$.}
    \label{fig:plots:l}
\end{figure}

This time, a sensible substitition is much more elusive. Figure \ref{fig:plots:l} showcases the effect on proper time of some typical parameters, including when acceleration is applied to increase the value of $L(r)$. Comparing with figure \ref{fig:plots:e}, we can see a much more dramatic shift in proper times at certain values of acceleration compared to the value of $L(r)$; this is because of how angular velocity changes depending on the radius.

We can also obtain an analytic expression for the 4-velocity as a function of $r$:
\begin{equation}
    \begin{split}
        u^{r} & = - \sqrt{ \left( 1 + \frac{\left[ L_0 + \alpha \left( \frac{1}{2} r (3M + r) \sqrt{\frac{2M}{r} - 1} + 3M^2 \arctan\sqrt{\frac{2M}{r} - 1}\right) \right]^2}{r^2}\right) \left(\frac{2M}{r} - 1 \right) } \\
        u^{\phi} & = \pm \frac{1}{r} \sqrt{ \left(\frac{2M}{r} - 1\right)(u^r)^2 - 1 } .
    \end{split}
\end{equation}

\begin{figure}[b]
    \includegraphics[clip,trim=0 10cm 0 0,width=0.7\linewidth]{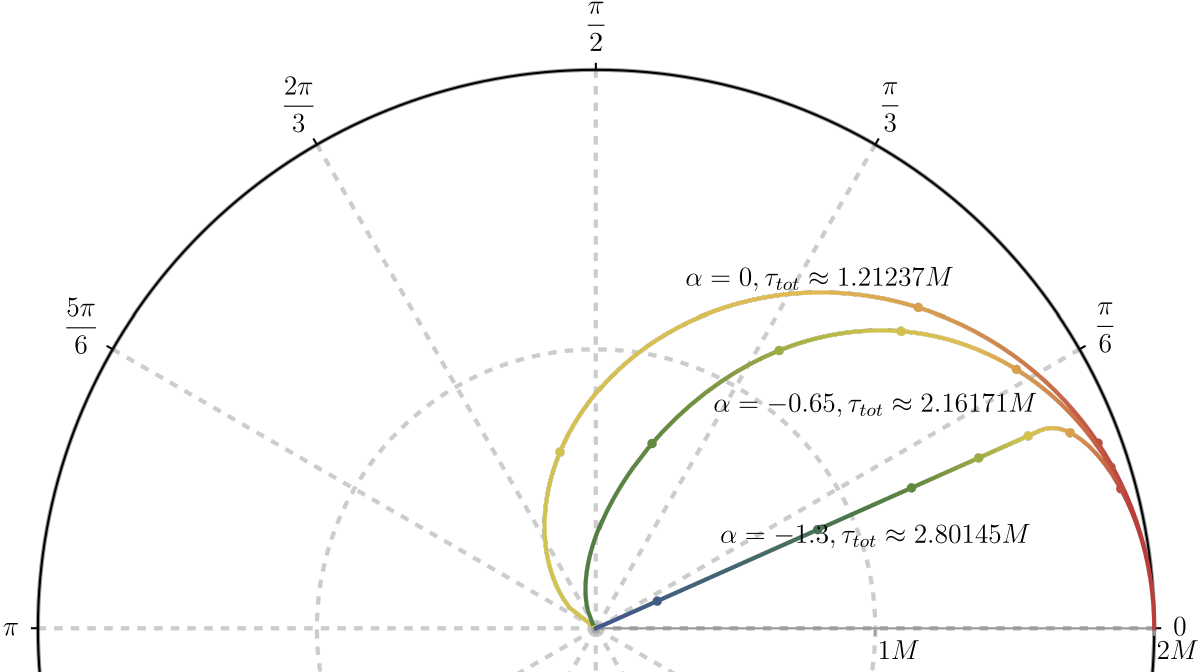}
    \caption{Curves of accelerated motion after infall with $L_0 = 3.5, e_0 = 0$ for various values of acceleration $a$. If we reach the optimal geodesic, we cease accelerating further, and this way, the more acceleration we have, the more we can live. Compare with \ref{fig:curves:l}.}
    \label{fig:curves:l_optimal}
\end{figure}      

\subsection{General case}

In the general case, there is a problem: we have an extra degree of freedom, corresponding to the angle at which we thrust with our engines. However, because of the curved geometry this angle is hard to define. Therefore, we will employ a geometric argument to simplify the problem, which will turn out to generalise to other spacetimes.

\begin{theorem}[Schwarzschild black hole survival time maximisation principle]\label{theorem:maximisation_principle}
    Assuming a rocket equipped astronaut falls towards the singularity with some 3-velocity $\vec{u}$ in relation to the optimal observer frame, to move on the worldline that maximises proper time, they must continuously thrust with all available engine power with a 3-acceleration $\vec{a}$ directed opposite to $\vec{u}$ until their 4-velocity matches \eqref{eq:u_opt}.
\end{theorem}
\begin{proof}
    Let us pick a point on our accelerated worldline $x^\mu = (r, t, \frac{\pi}{2}, \phi)$. At that point, there exists an instantenous inertial frame associated with the geodesic tangent to the worldline. In addition, there is a tangent 4-velocity $u^\mu$ and a perpendicular 4-acceleration $a^\mu$, which take the form $(1, 0, 0, 0)$ and $(0, a_t, 0, a_\phi)$ in that frame respectively.

    However, at the same point, there exists another valid inertial frame: the frame field spanned by observers with 4-velocity $u_{\textrm{opt}}^\mu = (- \sqrt{\frac{2M}{r} - 1}, 0, 0, 0)$. In this frame, the time component of the 4-velocity of the accelerated worldline turns out to be
    \begin{equation}
        u^r(\tau) = - \sqrt{\left(\frac{L(\tau)^2}{r^2} + 1\right) \left(\frac{2M}{r} - 1\right) + e(\tau)^2}.
    \end{equation}

    Since $r$ must necessarily decrease across any future-directed worldline, we know that the proper time of the optimal observer $\tau_{\textrm{opt}}$ can be described as a function of $r$. On an infinitesimal interval $(\tau_{\textrm{opt}}(r) - \epsilon, \tau_{\textrm{opt}} + \epsilon)$, we can treat the metric as Minkowskian in both these frames. Therefore, there exists a boost between the optimal frame and the worldline's frame. This is a Lorentz boost with a Lorentz factor
    \begin{equation}\label{eq:proof_lorentz_factor}
        \gamma = g_{\mu\nu} u^\mu u_{\textrm{opt}}^\nu = \sqrt{\frac{L^2}{r^2} + 1 + e^2 \left(\frac{2M}{r} - 1\right)^{-1}},
    \end{equation}
    and the ratio of proper time for the optimal worldline to the proper time for the accelerated worldline is 
    \begin{equation}
        \frac{d\tau_{\textrm{opt}}}{d\tau} = \gamma.
    \end{equation}

    We can perform this reasoning at enough points to fill the entire interval $\tau_{\textrm{opt}} \in (0, M \pi)$ (or equivalently, $r \in (0, 2M)$) and get 
    \begin{equation}\label{eq:tau_gamma}
        \tau = \int_{0}^{M \pi} \gamma^{-1} d\tau_{\textrm{opt}},
    \end{equation}
    which is clearly maximalised when $\gamma$ is minimised. Since the boost between the two frames was Lorentzian, this means that the Lorentz factor also obeys the definition
    \begin{equation}
        \gamma = \left(1 - |v|^2\right)^{-1/2}.
    \end{equation}

    Differentiating $\gamma$ in $\tau$ in the inertial frame gives
    \begin{equation}
        \frac{d\gamma}{d\tau} = \gamma^4 \vec{v} \cdot \vec{a},
    \end{equation}
    which becomes
    \begin{equation}
        \frac{d\gamma}{d\tau_{\textrm{opt}}} \frac{d\tau_{\textrm{opt}}}{d\tau} = \gamma^4 \vec{v} \cdot \vec{a},
    \end{equation}
    and finally 
    \begin{equation}\label{eq:dgamma_dr}
        \frac{d\gamma}{d\tau_{\textrm{opt}}} = \gamma^3 \vec{v} \cdot \vec{a}.
    \end{equation}

    To minimise $\gamma$ in \eqref{eq:tau_gamma} means the value of \eqref{eq:dgamma_dr} must be negative and have the highest absolute value possible, which happens when the spatial scalar product $\vec{v} \cdot \vec{a} = |\vec{v}| |\vec{a}| \cos{\pi}$, ie.\ when the 3-acceleration has a maximal magnitude and is directed exactly opposite to the 3-velocity in the inertial frame of the optimal observer.
\end{proof}

Another observation that can be made is that the algebraic expression \eqref{eq:proof_lorentz_factor} actually shows up in \eqref{eq:geodesic_proper_time}, and we can obtain
\begin{equation}
    \tau = \int_{0}^{2M} \gamma^{-1} \left( \frac{2M}{r} - 1 \right)^{-1/2} dr,
\end{equation}
from which an almost identical argument follows. However, we used the reasoning presented here because it neatly generalises to more general worldlines in other spacetimes. We shall comment on this in section \ref{sec:future:otherspacetimes}.

As far as calculation goes, one can perform a quadrature in which at every step we adjust the values of the parameters $L$ and $e$ in the optimal observer frame with \eqref{eq:dedr} and \eqref{eq:dLdr}, and then use the definition of the integral \eqref{eq:geodesic_proper_time} to sum proper times. However, we have to make sure we decompose the proper 3-acceleration correctly into the two components along the $\partial_t$ and $\partial_\phi$ axes -- as it was the proper acceleration that entered the equations \eqref{eq:dedr} and \eqref{eq:dLdr}.

First, we take the 3-acceleration with magnitude $\alpha$ from the momentary frame to the optimal observer frame. This boost will be parallel to the direction of the acceleration by theorem \ref{theorem:maximisation_principle}, so in the optimal observer's frame
\begin{equation}
    |a_{\textrm{opt}}| = \frac{\alpha}{\gamma^3},
\end{equation}
as is normal for a transformation of 3-acceleration with a Lorentz boost parallel to the vector's direction.

The 3-vector $\vec{a_{\textrm{opt}}}$ is then, by theorem \ref{theorem:maximisation_principle}
\begin{equation}\label{eq:avec_opt}
    \vec{a_{\textrm{opt}}} = - \frac{\alpha}{\gamma^3} \frac{\vec{u_{\textrm{opt}}}}{|u_{\textrm{opt}}|},
\end{equation}
where the velocity vector has been normalised and then the 3-acceleration has been aligned parallel but opposite to it.

3-acceleration transforms from the optimal observer frame to the instantaneous frame aligned with the accelerating observer as
\begin{equation}
    \vec{a} = \gamma^2 \left[ \vec{a_{\textrm{opt}}} + \frac{(\vec{a_{\textrm{opt}}} \cdot \vec{u_{\textrm{opt}}}) \vec{u_{\textrm{opt}}}}{|\vec{u_{\textrm{opt}}}|^2} (\gamma - 1)\right].
\end{equation}
Inserting \eqref{eq:avec_opt} into this equation we get
\begin{equation}
    \vec{a} = - \gamma^2 \frac{\alpha}{\gamma^3} \frac{\vec{u_{\textrm{opt}}}}{|u_{\textrm{opt}}|} \left[ 1 + \frac{(\vec{u_{\textrm{opt}}} \cdot \vec{u_{\textrm{opt}}}) }{|\vec{u_{\textrm{opt}}}|^2} (\gamma - 1)\right] = - \frac{\alpha}{\gamma} \frac{\vec{u_{\textrm{opt}}}}{|u_{\textrm{opt}}|} \left[ 1 + \gamma - 1\right] = - \alpha \frac{\vec{u_{\textrm{opt}}}}{|u_{\textrm{opt}}|}.
\end{equation}

Because of the alignment of the 3-velocity and 3-acceleration, everything cancels out, and we can simply obtain the components of proper 3-acceleration by getting the angle of $\vec{u_{\textrm{opt}}}$ with the $\partial_t$ and $\partial_\phi$ axes and multiply its' sine or cosine by the total proper acceleration -- exactly as if we were decomposing the acceleration in flat space.

\begin{figure}[p]
    \begin{subfigure}{0.49\textwidth}
        \includegraphics[width=0.9\linewidth]{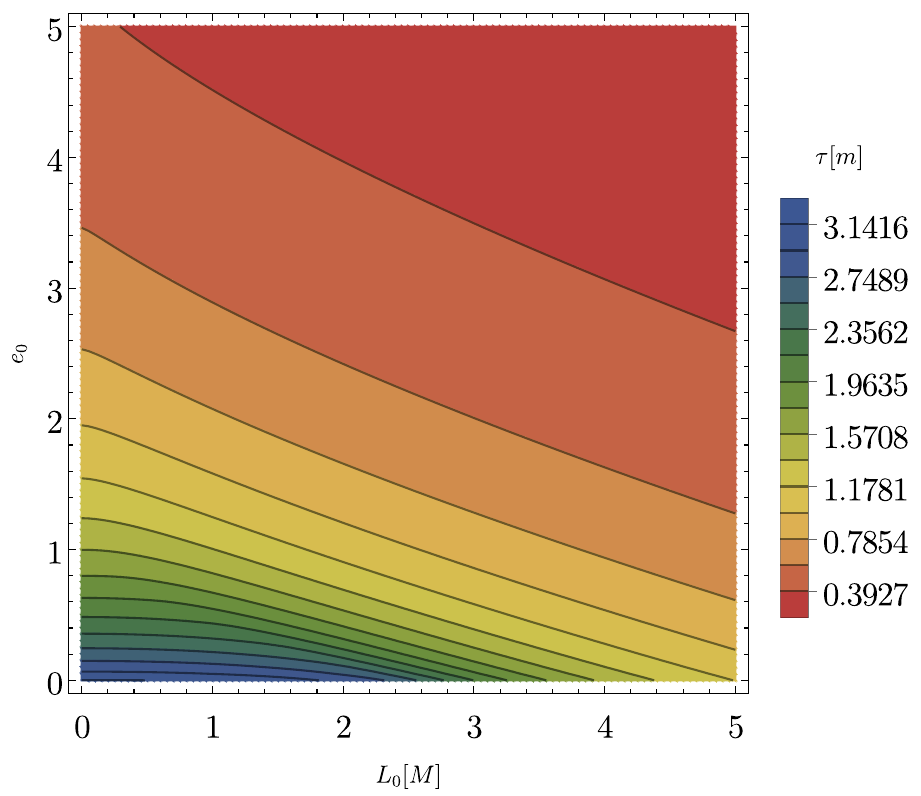} 
        \caption{$\alpha = -0.5 m^{-1}$, proper time $\tau$.}
    \end{subfigure}
    \begin{subfigure}{0.49\textwidth}
        \includegraphics[width=0.9\linewidth]{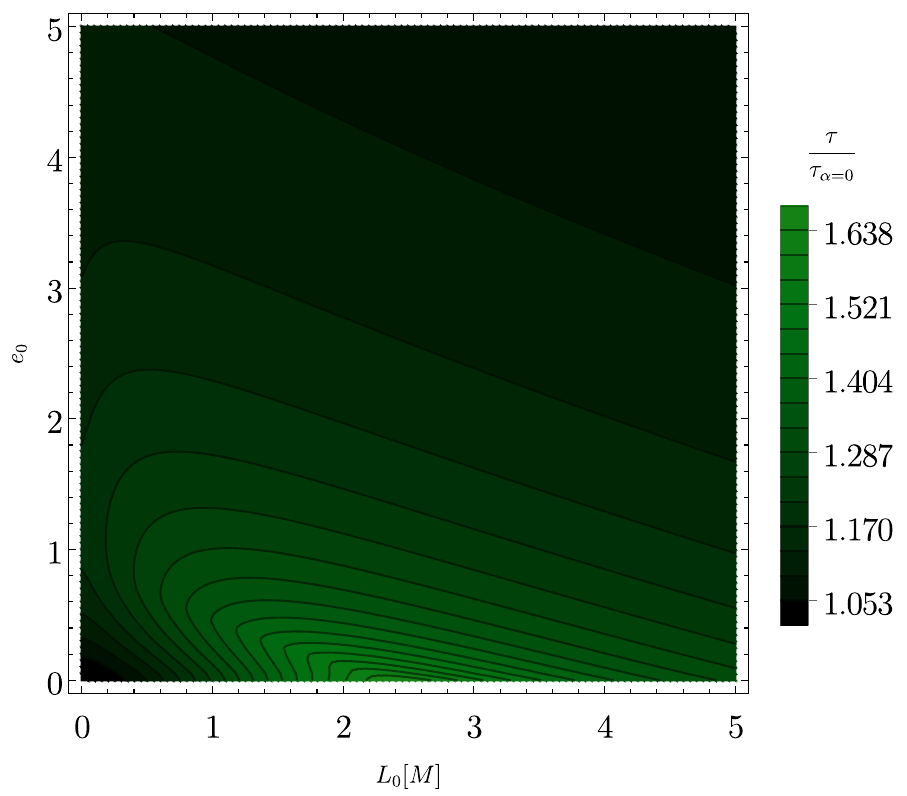} 
        \caption{$\alpha = -0.5 m^{-1}$, fraction of time gained $\frac{\tau}{\tau_{\alpha=0}}$.}
    \end{subfigure}

    \begin{subfigure}{0.49\textwidth}
        \includegraphics[width=0.9\linewidth]{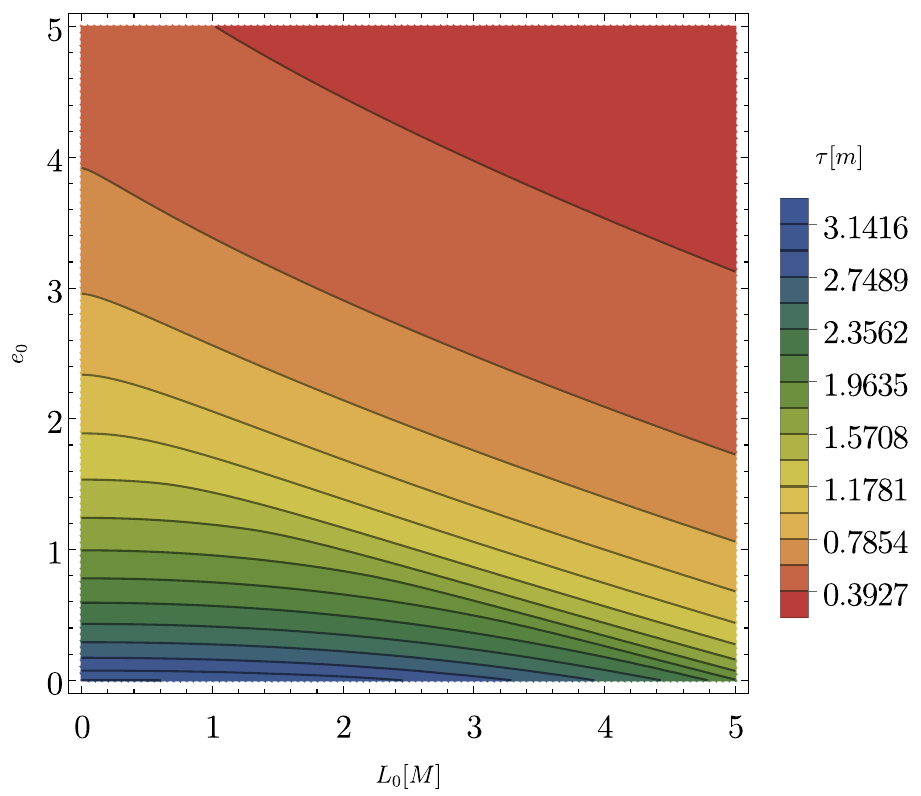} 
        \caption{$\alpha = -1.0 m^{-1}$, proper time $\tau$.}
    \end{subfigure}
    \begin{subfigure}{0.49\textwidth}
        \includegraphics[width=0.9\linewidth]{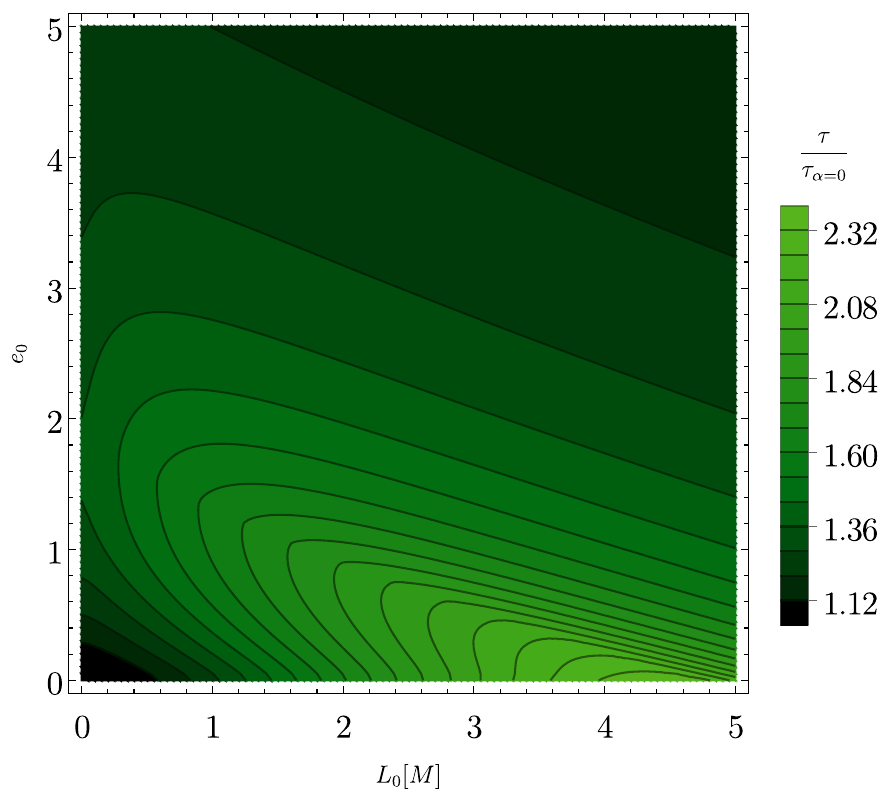} 
        \caption{$\alpha = -1.0 m^{-1}$, fraction of time gained $\frac{\tau}{\tau_{\alpha=0}}$.}
    \end{subfigure}
    
    \begin{subfigure}{0.49\textwidth}
        \includegraphics[width=0.9\linewidth]{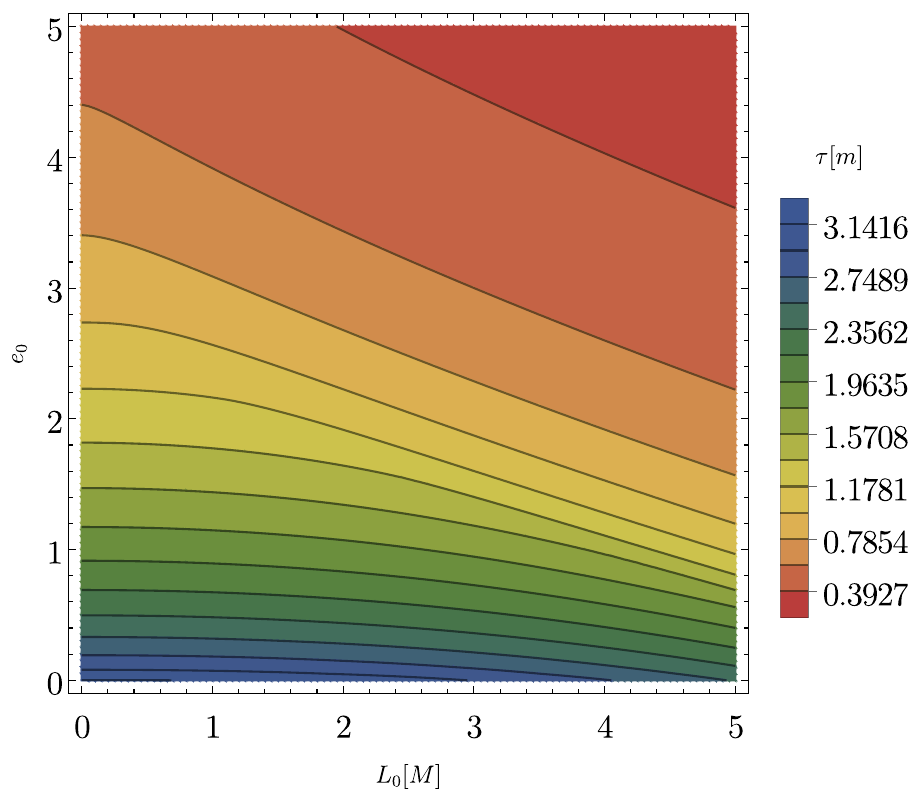} 
        \caption{$\alpha = -1.5 m^{-1}$, proper time $\tau$.}
    \end{subfigure}
    \begin{subfigure}{0.49\textwidth}
        \includegraphics[width=0.9\linewidth]{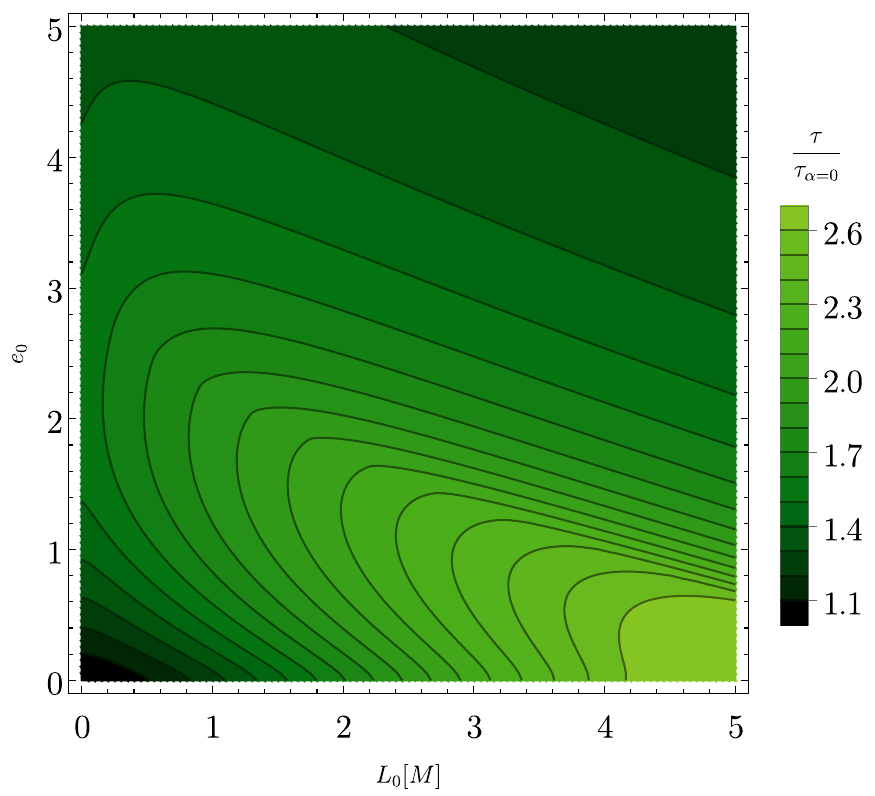} 
        \caption{$\alpha = -1.5 m^{-1}$, fraction of time gained $\frac{\tau}{\tau_{\alpha=0}}$.}
    \end{subfigure}
    \caption{Contour plots of proper time and time gained for infall depending on the parameters $e_0$ and $L_0$ for various values of $\alpha$, for a $M = 1m$ black hole. On the left are plots of proper time; colors retain meaning of proper time $\tau$ from worldline graphs; contours are separated by $\frac{\pi}{16}$. On the right are plots of time gained.}
    \label{fig:plots:general}
\end{figure}

\begin{figure}[p]
    \begin{subfigure}{0.49\textwidth}
        \includegraphics[width=0.9\linewidth]{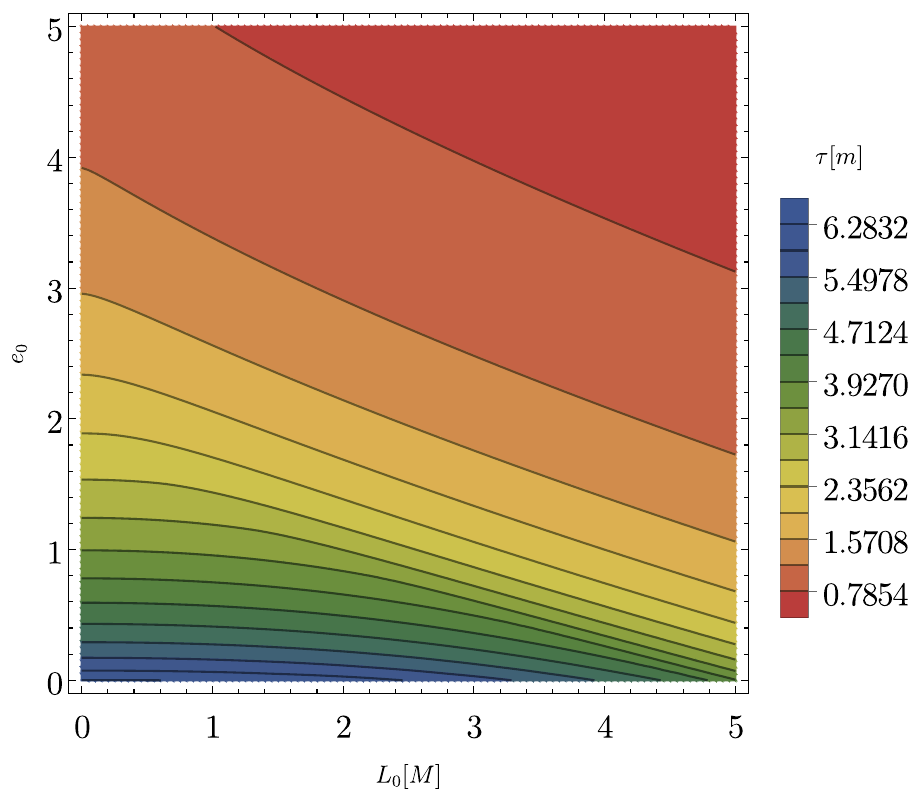} 
        \caption{$\alpha = -0.5 m^{-1}$, proper time $\tau$.}
    \end{subfigure}
    \begin{subfigure}{0.49\textwidth}
        \includegraphics[width=0.9\linewidth]{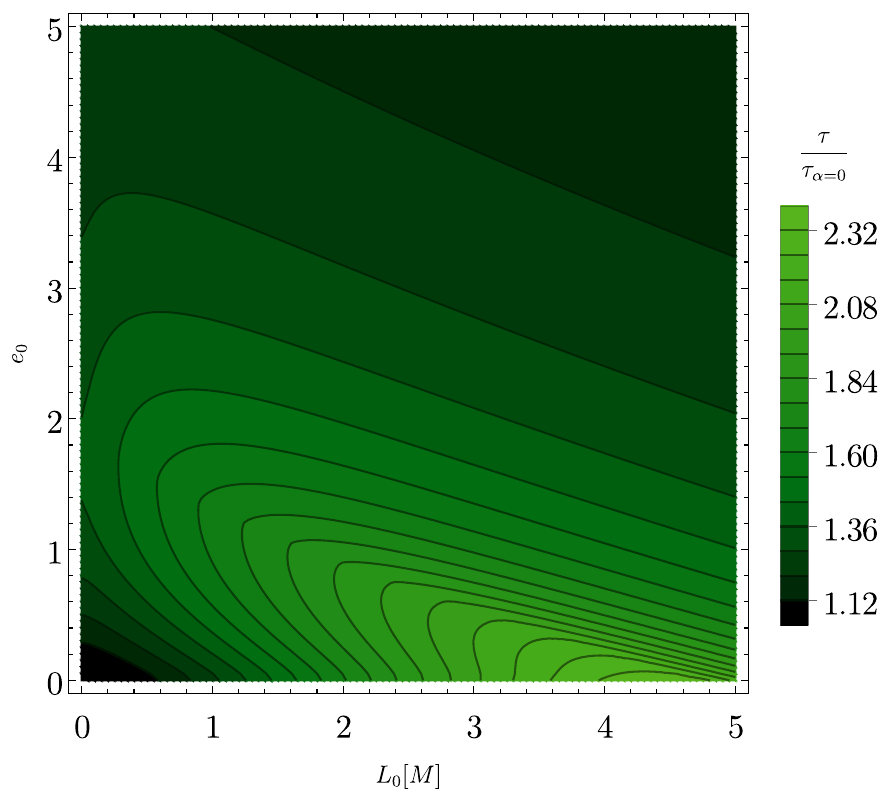} 
        \caption{$\alpha = -0.5 m^{-1}$, fraction of time gained $\frac{\tau}{\tau_{\alpha=0}}$.}
    \end{subfigure}

    \begin{subfigure}{0.49\textwidth}
        \includegraphics[width=0.9\linewidth]{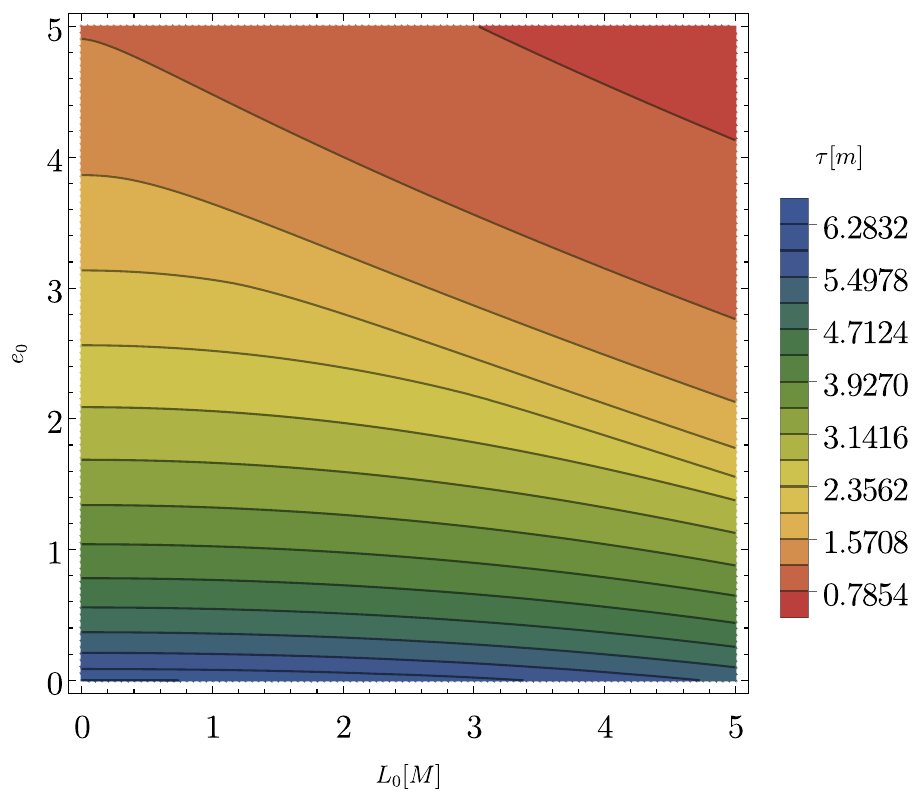} 
        \caption{$\alpha = -1.0 m^{-1}$, proper time $\tau$.}
    \end{subfigure}
    \begin{subfigure}{0.49\textwidth}
        \includegraphics[width=0.9\linewidth]{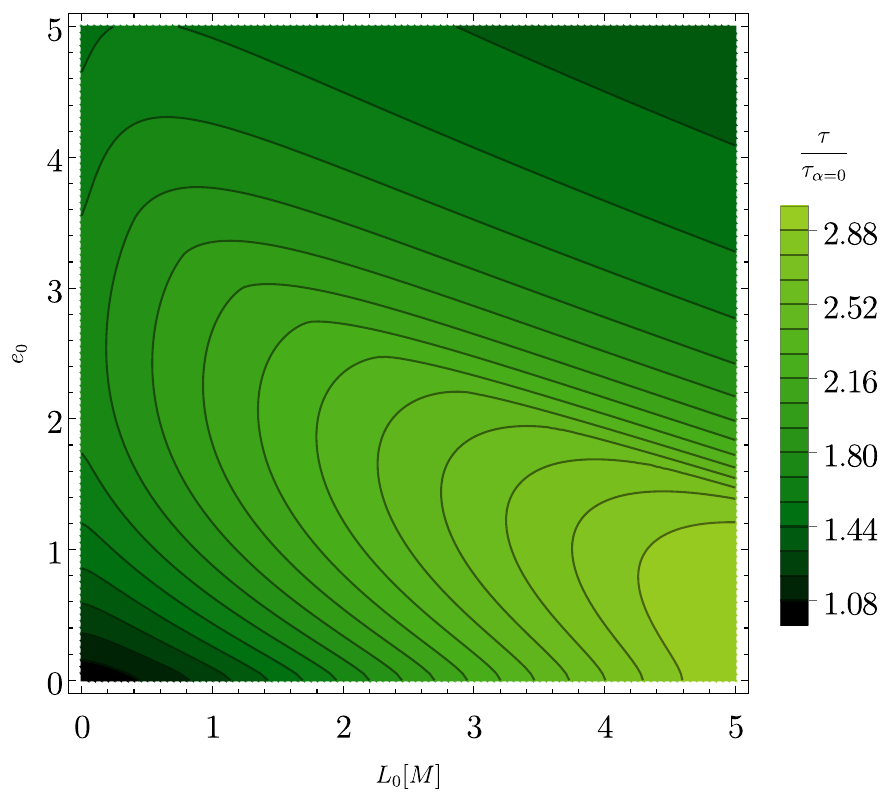} 
        \caption{$\alpha = -1.0 m^{-1}$, fraction of time gained $\frac{\tau}{\tau_{\alpha=0}}$.}
    \end{subfigure}
    
    \begin{subfigure}{0.49\textwidth}
        \includegraphics[width=0.9\linewidth]{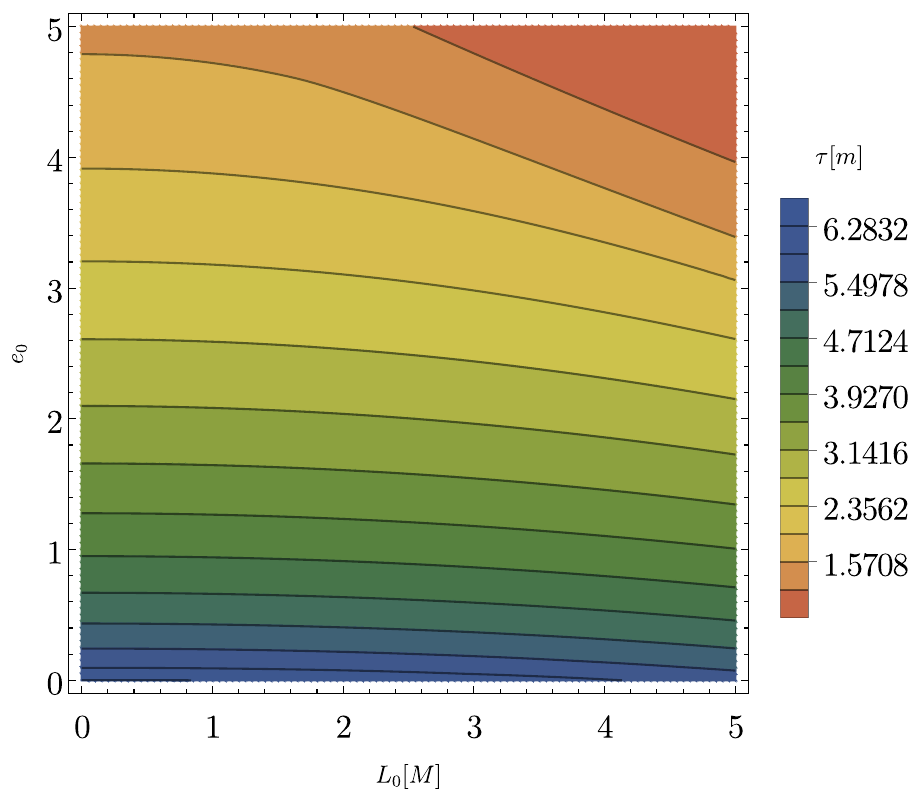} 
        \caption{$\alpha = -1.5 m^{-1}$, proper time $\tau$.}
    \end{subfigure}
    \begin{subfigure}{0.49\textwidth}
        \includegraphics[width=0.9\linewidth]{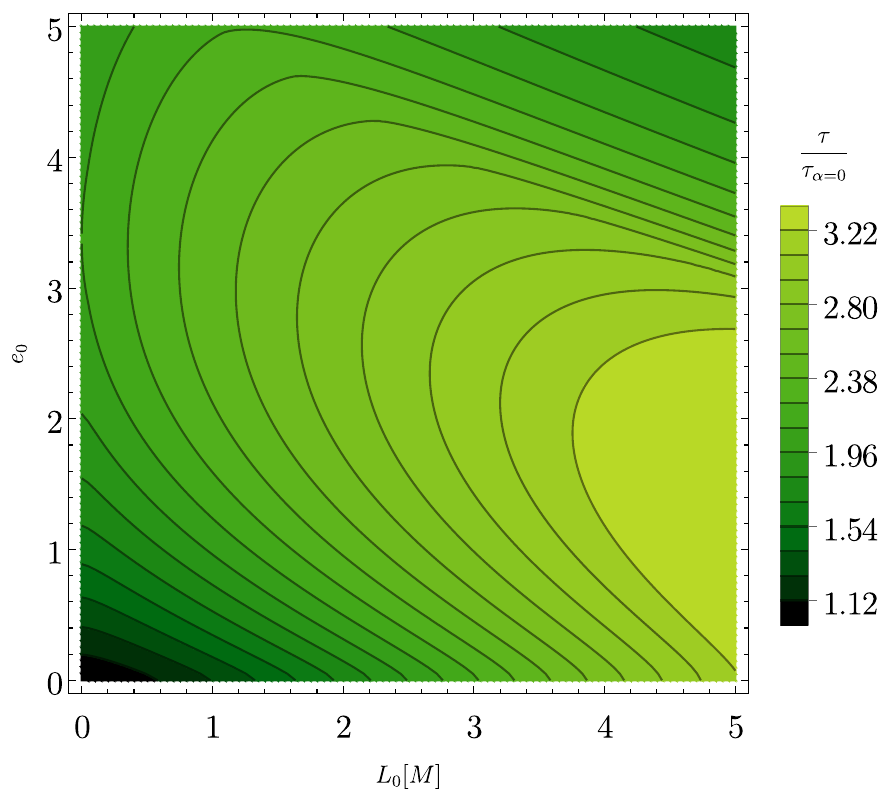} 
        \caption{$\alpha = -1.5 m^{-1}$, fraction of time gained $\frac{\tau}{\tau_{\alpha=0}}$.}
    \end{subfigure}
    \caption{Contour plots of proper time and time gained for infall depending on the parameters $e_0$ and $L_0$ for various values of $\alpha$, for a $M = 2m$ black hole. On the left are plots of proper time; colors retain meaning of proper time $\tau$ from worldline graphs; contours are separated by $\frac{\pi}{16}$. On the right are plots of time gained.}
    \label{fig:plots:general_2m}
\end{figure}

Figure \ref{fig:plots:general} showcases the proper time we can survive for a specific value of acceleration and black hole mass $1m$. Figure \ref{fig:plots:general_2m} showcases the proper time we can survive for a specific value of acceleration and black hole mass $2m$. The comparison of these figures, especially the case of $\alpha = 1 m^{-1}$ for the $1m$ black hole with $\alpha = 0.5 m^{-1}$ for the $2m$ black hole -- results in the same plot of the gain function. In other words, the behavior is dependent on the dimensionless product $M \alpha$, as indeed it must be. This has huge practical implications: the larger our black hole, the less engine power we need to effectively slow down.
\section{\label{sec:future}Results \& discussion}
\subsection{Gains for a `realistic' space traveller}
What would these results mean for astrophysical situations? We first need to determine what are the typical parameters we ought to consider.

For infall parameters $e_0$ and $L_0$, there are two `typical' scenarios. First one is a direct collision with the black hole with very little angular momentum; in this situation, we can use the analytical radial infall expression \eqref{eq:tau_r_carlsonj_to_opt} with $L_0 = 0$ and $e_0$ being either around $1$ (representing a free fall from far away) or higher (representing an unlucky colliding during some ultrarelativistic speed maneuver). The second viable scenario is some sort of accident when around the ISCO orbit, perhaps while performing research in a stable orbit around a black hole. In this scenario, $e_0 \approx 0.8$ and $L_0 = \sqrt{12} M$, as that is the angular momentum at which a potential barrier that makes infall unlikely shows up \cite{wald2010general}.

\begingroup
\squeezetable
\begin{table}[t]
\caption{Proper time before hitting the singularity depending on starting parameters}
\label{tab:tau_parameters}
\begin{tabular}{|l|l|l|l|l|}
    \hline
    mass M of black hole & starting parameters & $\alpha$  & $\tau$ & fraction of time gained $\frac{\tau}{\tau_{\alpha=0}} - 1$ \\ \hline
\multirow{4}{*}{$M = 4.297 * 10^6 M_\odot$ (SagA*)}    & \multirow{2}{*}{$e_0 = 1, L_0 = 0$}            & $0$     & $28.6467 s$ & N/A     \\ \cline{3-5} 
                                                       &                                                & $1000g$ & $28.6537 s$ & 0.00025 \\ \cline{2-5} 
                                                       & \multirow{2}{*}{$e_0 = 0.8, L_0 = \sqrt{12}M$} & $0$     & $16.156 s$  & N/A     \\ \cline{3-5} 
                                                       &                                                & $1000g$ & $16.160 s$  & 0.00025 \\ \cline{1-5} 
\multirow{10}{*}{$M = 6.5 * 10^9 M_\odot$ (M87*)}       & \multirow{5}{*}{$e_0 = 1, L_0 = 0$}           & $0$     & $12.037 h$  & N/A     \\ \cline{3-5} 
                                                       &                                                & $1g$    & $12.042 h$  & 0.00037 \\ \cline{3-5} 
                                                       &                                                & $10g$   & $12.082 h$  & 0.00372 \\ \cline{3-5} 
                                                       &                                                & $100g$  & $12.491 h$  & 0.0377 \\ \cline{3-5} 
                                                       &                                                & $1000g$ & $16.139 h$  & 0.341 \\ \cline{2-5} 
                                                       & \multirow{5}{*}{$e_0 = 0.8, L_0 = \sqrt{12}M$} & $0$     & $ 6.788 h$  & N/A     \\ \cline{3-5} 
                                                       &                                                & $1g$    & $ 6.790 h$  & 0.00039 \\ \cline{3-5} 
                                                       &                                                & $10g$   & $ 6.815 h$  & 0.00389 \\ \cline{3-5} 
                                                       &                                                & $100g$  & $ 7.065 h$  & 0.0407  \\ \cline{3-5} 
                                                       &                                                & $1000g$ & $12.537 h$  & 0.847 \\ \cline{1-5} 
\multirow{10}{*}{$M = 4.07 * 10^{10} M_\odot$ (TON618)}& \multirow{5}{*}{$e_0 = 1, L_0 = 0$}            & $0$     & $3.1404 d$  & N/A     \\ \cline{3-5} 
                                                       &                                                & $1g$    & $3.1477 d$  & 0.00233 \\ \cline{3-5} 
                                                       &                                                & $10g$   & $3.2142 d$  & 0.0235 \\ \cline{3-5} 
                                                       &                                                & $100g$  & $3.8778 d$  & 0.235 \\ \cline{3-5} 
                                                       &                                                & $1000g$ & $5.6014 d$  & 0.784 \\ \cline{2-5} 
                                                       & \multirow{5}{*}{$e_0 = 0.8, L_0 = \sqrt{12}M$} & $0$     & $1.771 d$  & N/A     \\ \cline{3-5} 
                                                       &                                                & $1g$    & $1.775 d$  & 0.00243 \\ \cline{3-5} 
                                                       &                                                & $10g$   & $1.815 d$  & 0.0251 \\ \cline{3-5} 
                                                       &                                                & $100g$  & $2.411 d$  & 0.361  \\ \cline{3-5} 
                                                       &                                                & $1000g$ & $5.773 d$  & 2.260 \\ \cline{1-5} 
\multirow{10}{*}{$M = 10^{11} M_\odot$ (PhoenixA)}        & \multirow{5}{*}{$e_0 = 1, L_0 = 0$}            & $0$     & $7.7160 d$  & N/A     \\ \cline{3-5} 
                                                       &                                                & $1g$    & $7.7603 d$  & 0.00573 \\ \cline{3-5} 
                                                       &                                                & $10g$   & $8.1671 d$  & 0.0585 \\ \cline{3-5} 
                                                       &                                                & $100g$  & $11.173 d$  & 0.448 \\ \cline{3-5} 
                                                       &                                                & $1000g$ & $15.065 d$  & 0.952 \\ \cline{2-5} 
                                                       & \multirow{5}{*}{$e_0 = 0.8, L_0 = \sqrt{12}M$} & $0$     & $4.352 d$  & N/A     \\ \cline{3-5} 
                                                       &                                                & $1g$    & $4.378 d$  & 0.00601 \\ \cline{3-5} 
                                                       &                                                & $10g$   & $4.633 d$  & 0.0646 \\ \cline{3-5} 
                                                       &                                                & $100g$  & $11.008 d$  & 1.530  \\ \cline{3-5} 
                                                       &                                                & $1000g$ & $15.407 d$  & 2.541 \\ \hline
\end{tabular}
\end{table}
\endgroup

For the mass of the black hole $M$, we can take a few representative examples: the Sagittarius A* black hole $M = 4.297 * 10^6 M_\odot$ \cite{2023sagittariusa}, Messier 87* $M = 6.5 * 10^9 M_\odot$ \cite{2019messier87}, TON 618 $M = 4.07 * 10^{10} M_\odot$ \cite{2019ton618}, and Phoenix A $M = 10^{11} M_\odot$ \cite{2016phoenixa}. Naturally, real astrophysical black holes spin, and we are assuming a spherical symmetry in this work. Nevertheless, we can obtain some instructive values.

That leaves the final parameter, the acceleration. In principle, one could hypothesize extreme acceleration a far-future engine could be capable of. However, there is a physiological limit of acceleration a human can handle. A sustained acceleration of $1g \approx 10 \frac{m}{s^2}$ would mimic gravity on the surface of the Earth and thus be very comfortable for passengers. Trained pilots can pull around $5g$ of force -- this would be a very unpleasant trip. Around that limit, the pressure differential caused by the acceleration in the cardiovascular system precludes long term survival in air; however, when submerged in a rigid tank filled with water, a human can survive up to $24g$ in relative comfort \cite{rossini2007beyond}. At that point, the air-filled chest cavity (known colloquially as `lungs') threatens collapse. It is still theoretically possible to accelerate faster, however, as one could also use some sort of breathable liquid \cite{gabriel1996quantitative}. Perfluorocarbons that can do this exist, however, one would ideally need a fluid of density very close to water, to match the surrounding tissues. While non-existent as of now, it is a realistic technological advance. With it, accelerations of $100g$ and even more are theorised to be survivable \cite{rossini2007beyond}.

We will settle on $a = 1g$, $a = 10g$, $a = 100g$ and $a = 1000g$ as our `realistic' parameters, with the understanding that the highest value is probably the absolute ceiling of what a human could coinceivably survive and the second highest requires large technological advances to be viable.

A table of results for each black hole is presented in \ref{tab:tau_parameters}.

We notice immediately a few interesting things. First off, when mass of the black hole is very low, accelerations on the order of $g$ have almost no effect on our survival time. However, for extremely large supermassive black holes, the effects are much higher: struggling is already a good idea when we're faced with the prospect of falling into Messier 87*. For even larger black holes, $10g$ -- completely achievable with current technology -- can give us an extra 6 percent more time, and a hypothetical $1000g$ can double or triple our survival time. Of course, we are always limited by the $\pi M$ upper bound. Nevertheless, our gains can be significant. To quote Dylan Thomas \cite{thomas2003poems}: ,,\textit{Do not go gentle into that good night}.''

\section{\label{sec:future:otherspacetimes}Generalised maximal survival principle}

The argument presented in theorem \ref{theorem:maximisation_principle} generalises in a straightforward manner. We will start with a geometrical construction of a particular vector field.

\begin{lemma}[Sufficient conditions for the existence of an optimal observer frame field between two hypersurfaces]\label{theorem:optimal_observer_field_existence}
    Let $M$ be a Lorentzian manifold of signature $(1, n-1)$. On this manifold, let there exist two distinct spatial hypersurfaces $S_1, S_2$ (foliations) such that every future-directed timelike geodesic starting at $S_1$ must reach $S_2$, and every past-directed timelike geodesic starting at $S_2$ must reach $S_1$.
    If there exist $n-1$ Killing spacelike vector fields $K_i$ such that every $K_i$ commutes with $K_j$ $\forall i \ne j$ at every point of $M$ between $S_1$ and $S_2$, and the hypersurfaces $S_1$ and $S_2$ are aligned with $K_i$, then there exists a geodesic vector field describing n-velocities of massive observers who survive a universal maximal possible amount of proper time between the foliations $S_1$ and $S_2$ (among all possible timelike observers). This field spans the entire region of the manifold between these foliations, and the n-velocity $u^\mu$ describing this field is purely timelike along its entire length.
\end{lemma}
\begin{proof}
    Without loss of generality, we begin by assuming we are endowed with a metric 
    \begin{equation}\label{eq:general_metric_line_element}
        ds^2 = - g_{tt}(t)dt^2 + \sum_{i=1}^n g_{ii}(t) (dx^i)^2,
    \end{equation}
    where $\partial_t$ is a purely timelike vector, and $\partial_i$ are $n-1$ spacelike vectors aligned such that $\partial_i$ is parallel to $K_i$. The metric components can thus only depend on the $t$ coordinate.

    In this metric, each Killing vector $K_i$ admits a constant of geodesic motion (the Einstein summation convention not applying for $i$ in the following equations)
    \begin{equation}\label{eq:constant_general}
        C_i = u^\mu (K_i)_\mu = g_{ii}(t) u^i (K_i)^i,
    \end{equation}
    which after assuming normalisation of the Killing vectors gives
    \begin{equation}\label{eq:u_i_general}
        u^i = \frac{C_i}{g_{ii}(t)}.
    \end{equation}

    Thus, our n-velocity is completely determined as (from normalisation $u^\mu u_\mu = -1$)
    \begin{equation}
        u^t = \sqrt{g^{-1}_{tt}(t)} \sqrt{1+ \sum_{i=1}^n C_i^2 g^{-1}_{ii}(t)}.
    \end{equation}

    In this metric, the proper time can be calculated as
    \begin{equation}
        \tau = \int_{\tau_0}^{\tau_{\textrm{max}}} d\tau = \int_{\tau_0}^{\tau_{\textrm{max}}} d\tau \frac{dt}{dt} = \int_{\tau_0}^{\tau_{\textrm{max}}} \frac{dt \sqrt{g_{tt}(t)}}{\sqrt{1+ \sum_{i=1}^n C_i^2 g^{-1}_{ii}(t)}},
    \end{equation}
    where $\tau_0$ and $\tau_{\textrm{max}}$ correspond to points on $S_1$ and $S_2$ respectively.

    This is maximalised when $\forall i: C_i = 0$, from which we know that
    \begin{equation}
        u^\mu = (\sqrt{g^{-1}_{tt}(t)}, \vec{0}).
    \end{equation}

    This is a timelike geodesic, and by symmetries exactly one unique worldline described by this geodesic passes through every point in the region between $S_1$ and $S_2$ on the manifold $M$. 
\end{proof}

Stated simply, the lemma \ref{theorem:optimal_observer_field_existence} shows that if we have $n-1$ pointwise commuting Killing vector fields, then the entire region of spacetime is spanned with timelike geodesics of optimal observers. They are parallel in the sense of parallel transport along spatial hypersurfaces. They never cross, and have equal total length between the foliations $S_1$ and $S_2$.

Without these symmetries, one can imagine a situation where the region of the manifold is not fully covered by optimal observer geodesics; or, equally disastrously for any attempts at figuring out how to use an engine, nearby optimal geodesics may not universally have the same elapsed proper time.

With a frame field constructed this way, we can generalise the result for maximising time using rockets.

\begin{theorem}[Generalised principle for delaying an inevitable cataclysmic event]\label{theorem:principle_of_cataclysmic_survival}
    Let $M$ be a Lorentzian $(1, n-1)$ manifold, and $P$ a point on the $n-1$ dimensional spatial hypersurface $S_1$ on this manifold. Let $S_2$ be another $n-1$ dimensional spatial hypersurface such that it must be reached by every future-directed worldline starting from hypersurface $S_1$ (and every past-directed worldline starting on $S_2$ must reach $S_1$). Let there exist $n-1$ spacelike commuting Killing vector fields at every point in the region between $S_1$ and $S_2$.

    Then the proper time $\tau_{\textrm{a}}$ for the rocket equipped observer starting at point $P$ is maximalised by thrusting with their engines directly opposite the 3-velocity in relation to the optimal observer frame field from lemma \ref{theorem:optimal_observer_field_existence}.
\end{theorem}
\begin{proof}
    By the construction from lemma \ref{theorem:optimal_observer_field_existence}, we have a field of parallel geodesics that maximalise time between $S_1$ and $S_2$.

    Thanks to the spatial symmetries, we can parallel transport each instantaneous frame of the accelerated observer along spatial axes, and thus reduce the problem to a one dimensional integral along a single optimal observer geodesic with no ambiguity:
    \begin{equation}
        \tau_{\textrm{a}} = \int_{0}^{\tau} \gamma^{-1} d\tau,
    \end{equation}
    where the $\gamma$ parameter is a Lorentz boost factor between the instantaneous frame of the accelerated observer and the frame of the optimal observer. 

    Henceforth, the proof proceeds exactly as for the Schwarzschild case in theorem \ref{theorem:maximisation_principle}.
\end{proof}

\begin{corollary}[Induced submanifold]
    The principle works for worldlines constrained to an induced submanifold that fulfills the requirements of theorem \ref{theorem:principle_of_cataclysmic_survival}.
\end{corollary}

This means we can use it when we have some constraints on our motion, as long as the submanifold of the constrained motion has a Killing vector field along each spatial dimension. In fact, this is precisely what we did in the Schwarzschild case -- we first used spherical symmetry to restrict motion to the equator, at which point all the remaining Killing vector fields are spacelike and commute with each other.

\begin{corollary}[Locality]\label{corollary:principle_of_cataclysmic_survival_locality}
    The same applies when the hypersurface is cataclysmic only for a subset of the manifold, as long as the entire timelike future of the rocket equipped observer event is covered by the optimal frame field.
\end{corollary}

\begin{corollary}[]\label{corollary:principle_of_cataclysmic_survival_past}
    A similar argument applies for the timelike past of an observer.
\end{corollary}

In the case where the metric does take the form \eqref{eq:general_metric_line_element}, we can also quite easily obtain a general expression for $\frac{dC_i}{dt}$.

\begin{theorem}[Rate of change of $C_i$]
    When the metric takes the form \eqref{eq:general_metric_line_element} and the n-velocity along worldline $u^\mu$ has only one non-zero spacelike component \eqref{eq:u_i_general}, then the rate of change of $C_i$ \eqref{eq:constant_general} with respect to the timelike coordinate $t$ under proper acceleration $\alpha$ is
    \begin{equation}
        \frac{dC_i}{dt} = - \alpha \sqrt{g_{ii}} \sqrt{g_{tt}^{-1}}.
    \end{equation}
\end{theorem}
\begin{proof}
    Just as in the case of the derivation of \eqref{eq:dedr} or \eqref{eq:dLdr}, we take the absolute derivative and use \eqref{eq:gen_eom}, \eqref{eq:a_norm}, \eqref{eq:u_norm} and \eqref{eq:a_orthogonal} to represent $a^i$ in terms of $\alpha$ and $u^i$.
\end{proof}

There are quite a few interesting spacetimes that fulfill the requirements. For instance, in the Big Bang model assuming a cataclysm at some $t_{\textrm{cataclysm}}$ such as a Big Crunch or a Big Rip, there is a class of observers going along the Hubble flow that spans the entire spacetime \cite{carroll2004spacetime}. This means survival is maximised the same way. So, a civilization on a rogue planet or colony ship moving quickly in relation to the Hubble flow would desire to slow their speed down when the universe is approaching its' end.

We can also work backwards -- for instance, towards the naked singularity of the Big Bang. The `oldest' structures in the universe are those aligned with the Hubble flow; objects that move fast in relation to it may be `younger'.

It's also worth noting that the `cataclysm' need not necessarily be cataclysmic: it can simply be a hypersurface the crossing of which is inevitable. As long as the entire timelike future up to the hypersurface can be covered with the frame field with necessary symmetries, the principle applies. This is a valuable geometrical tool in situations where worldlines are accelerated, and where we don't have a single definite initial/terminal event in spacetime. In these circumstances, geodesic maximisation of proper time need not necessarily apply, as illustrated by the Schwarzschild example discussed in this paper. However, a field of optimal observers may be possible to construct, and thus we obtain a geometric tool of analysis.

For instance, the Reissner–Nordström spacetime constrained to equatorial motion fulfills these requirements between the inner and outer horizon \cite{carroll2004spacetime}. Similarly, a chosen spatial foliation $t = const$ in Minkowski spacetime or a chosen cosmological time $t = const$ in the Big Bang model can be the target hypersurface. The Bianchi type I cosmological models and the Kasner metric \cite{misner1973gravitation} also fulfill these requirements.

\subsection{Future avenues of research}

The principle in theorem \ref{theorem:principle_of_cataclysmic_survival} is a valuable tool in various other spacetimes, some of which have been elucidated in this paper.

One could also pose a similar question for the Kerr spacetime. While in the Kerr spacetime there exist closed timelike loops which allow arbitrary survival time mathematically, it is all but certain that the inner horizon of the Kerr spacetime is a boundary beyond which general relativity breaks down. One could, therefore, try to find what a rocket-equipped astronaut can do to extend their proper time between the outer horizon and the inner horizon. Since Kerr spacetime exhibits fewer symmetries, the behaviour can be expected to be much richer. Also, Kerr black holes are actual astrophysical objects, making the problem somewhat more practical.

\begin{acknowledgements}

The author would like to thank Adam Cieślik and Sebastian Szybka for their valuable contributions in discussions about this problem, and for their patience in listening to a physicist desperately avoiding doing algebra.

The author would also like to extend gratitude to Tristan Needham for writing the books \textit{Visual Complex Analysis} and \textit{Visual Differential Geometry and Forms} \cite{needhamvca, needhamvdgaf}. The idiosyncratic geometrical style of solving problems presented therein was the catalyst that impelled this investigation, and inspired some of the proofs.

\end{acknowledgements}

\bibliographystyle{siam}
\bibliography{bibliography}

\end{document}